\documentclass[12pt,oneside,reqno,frenchspacing]{amsart}

\usepackage{amsfonts, amssymb, amsthm, amsmath, array, authblk, bm, booktabs, braket, caption, changepage, dsfont, enumitem, etoolbox, fullpage, graphicx, mathrsfs, mathtools, microtype, multicol, subfig, subfloat, textcomp, tikz, tikz-cd, url, xpatch, natbib}

\usepackage{hyperref}
\hypersetup{
	colorlinks=true,
	linkcolor=blue,
	citecolor=black,
	urlcolor=cyan}
\usepackage{cleveref}

\setlist[enumerate]{itemsep=0pt, topsep=0pt}
\setlist[itemize]{itemsep=0pt, topsep=0pt}

\usepackage[foot]{amsaddr}

\usepackage{newpxtext} \usepackage[euler-digits]{eulervm}

\usepackage{titlesec}

\titleformat{\section}
	{\titlerule\vspace{-2ex}}
	{\thesection.}{1ex}{\centering\scshape}
	[\vspace{.5ex}\titlerule]
	
\titleformat{\subsection}[runin]
  {\normalfont\normalsize\bfseries}{\thesubsection.}{1ex}{}	
  
\titlespacing*{\subsection}{0pt}{0.0\baselineskip}{1ex}

\newtheorem{theorem}{Theorem}[section]

\newtheoremstyle{style}
  {\baselineskip} 
  {0em} 
  {\itshape} 
  {} 
  {\bfseries} 
  {.} 
  {.5em} 
  {} 
\theoremstyle{style}

\newtheorem*{theorem*}{Theorem}
\numberwithin{equation}{section}
\newtheorem{assumption}{Assumption}

\newtheorem{definition}{Definition}
\newtheorem*{definition*}{Definition}
\newtheorem{example}[theorem]{Example}

\newtheorem{lemma}[theorem]{Lemma}
\newtheorem*{lemma*}{Lemma}
\newtheorem{proposition}[theorem]{Proposition}

\makeatletter
\renewenvironment{proof}[1][\proofname]{\par
  \pushQED{\qed}%
  \normalfont
  \topsep0ex   
  \trivlist
  \item[\hskip\labelsep\itshape
    #1\@addpunct{.}]\ignorespaces
}{%
  \popQED\endtrivlist
  \pushQED{}%
}
\makeatother

\usepackage[ruled]{algorithm2e}

\makeatletter
\renewcommand{\algocf@captiontext}[2]{#1\algocf@typo. \AlCapFnt{}#2} 
\def\@algocf@capt@plain{top}
\renewcommand{\algocf@makecaption}[2]{%
  \addtolength{\hsize}{\algomargin}%
  \sbox\@tempboxa{\algocf@captiontext{#1}{#2}}%
  \ifdim\wd\@tempboxa >\hsize
    \hskip .5\algomargin%
    \parbox[t]{\hsize}{\algocf@captiontext{#1}{#2}}
  \else%
    \global\@minipagefalse%
    \hbox to\hsize{\box\@tempboxa}
  \fi%
  \addtolength{\hsize}{-\algomargin}%
}
\makeatother

\DeclareMathOperator*{\argmax}{arg\,max}

\title{Sequential Gibbs Posteriors with Applications to Principal Component Analysis}

\author{Steven Winter\textsuperscript{1}}
\author{Omar Melikechi\textsuperscript{1}}
\author{David B. Dunson\textsuperscript{1,3}}
\address{\textsuperscript{1}Department of Statistical Science, Duke University, Durham, North Carolina} 
\address{\textsuperscript{3}Department of Mathematics, Duke University, Durham, North Carolina}

\begin{document}

\frenchspacing

\maketitle

\begin{abstract}
Gibbs posteriors are proportional to a prior distribution multiplied by an exponentiated loss function, with a key tuning parameter weighting information in the loss relative to the prior and providing control of posterior uncertainty. Gibbs posteriors provide a principled framework for likelihood-free Bayesian inference, but in many situations, including a single tuning parameter inevitably leads to poor uncertainty quantification. In particular, regardless of the value of the parameter, credible regions have far from the nominal frequentist coverage even in large samples. We propose a sequential extension to Gibbs posteriors to address this problem. We prove the proposed sequential posterior exhibits concentration and a Bernstein-von Mises theorem, which holds under easy to verify conditions in Euclidean space and on manifolds. As a byproduct, we obtain the first Bernstein-von Mises theorem for traditional likelihood-based Bayesian posteriors on manifolds. All methods are illustrated with an application to principal component analysis.
\end{abstract}
The standard Bayesian approach to data analysis involves specifying a generative model for the data via the likelihood, defining priors for all parameters, and computing parameter summaries using the posterior distribution defined by Bayes' rule. This paradigm has a number of advantages, allowing rich hierarchical models for complicated data generating processes, inclusion of expert information, and a full characterization of uncertainty in inference. One practical challenge arises in specifying realistic likelihoods for complex, high-dimensional data such as images or spatiotemporal processes. Realistic likelihoods from highly flexible parametric families may depend on more parameters than can be estimated from the data, introducing both theoretical and practical challenges for Bayesian analysis. Conversely, tractable likelihoods may miss important aspects of the data generating mechanism, leading to bias in posterior estimates, under-representation of parameter uncertainty, and poor predictive performance. The goal of this article is to extend likelihood-free Bayesian inference by leveraging loss-based learning.

Loss-based learning is an alternative approach which typically defines a loss measuring how well a parameter describes the data, estimates parameters by minimizing the loss, and occasionally quantifies estimation uncertainty relying on distributional assumptions, large-sample asymptotics, or nonparametric methods such as the bootstrap. This paradigm avoids specification of a likelihood, sidestepping the unfavorable trade-off between realism and tractability, but often requires strong distribution assumptions or large sample sizes for valid characterization of uncertainty \citep{ogasawara2002concise, potscher2009distribution, van2014asymptotically, zhang2014confidence}. Nonparametric methods such as the bootstrap perform well in a wide array of situations, but may under-represent uncertainty when data are heavy tailed, contain outliers, or are high dimensional \citep{schenker1985qualms, hall1990asymptotic, kysely2008cautionary, karoui2016can}. 

Gibbs posteriors offer an appealing middle-ground between the Bayesian and loss-based paradigms by replacing the negative log-likelihood with a loss function. Given a loss $\ell^{(n)}$ linking a parameter $\theta$ to $n$ observations $ x=(x_1,\dots,x_n)$, inference is based on the Gibbs posterior,
\begin{equation}
    \Pi^{(n)}_{\eta}(d\theta\mid  x) \propto \exp\{-\eta n\ell^{(n)}(\theta\mid x)\}\Pi^{(0)}(d\theta),\label{eq:gibbs}
\end{equation}
where $\Pi^{(0)}$ is the prior and $\eta>0$ is a hyperparameter weighting information in the loss relative to the prior. Gibbs posteriors allow valid Bayesian inference on $\theta$ without needing to specify a likelihood function. Furthermore, probabilistic statements from \eqref{eq:gibbs} do not rely on distributional assumptions, large-sample asymptotics, or non-parametric approximations. The robustness properties of \eqref{eq:gibbs} have been exploited in applications such as logistic regression, quantile estimation, image boundary detection, and clustering \citep{jiang2008gibbs, syring2020robust, bhattacharya2022gibbs, rigon2023generalized}.

Gibbs posteriors can be justified from a variety of foundational perspectives.
\cite{bissiri2016general} begin with the goal of updating prior beliefs about a risk minimizer, and derive \eqref{eq:gibbs} as the unique, coherent generalization of Bayes' rule. This provides rigorous justification for use of \eqref{eq:gibbs} in Bayesian inference. Gibbs posteriors also arise when studying the generalization error of randomized algorithms \citep{guedj2019primer}. A common goal in this literature is to establish high-probability upper bounds on the risk or average risk of a randomized estimator; \eqref{eq:gibbs} is obtained by minimizing an upper bound for the average risk \citep{zhang2006a, zhang2006b}. In this context, the loss is an empirical risk function and the prior is an arbitrary reference measure. This framework has been used to provide new insights into classical methods such as empirical risk minimization, and to derive state of the art generalization bounds for modern machine learning algorithms \citep{dziugaite2017computing, neyshabur2017exploring}.

In practice, the performance of \eqref{eq:gibbs} depends critically on $\eta$, which appears because the scale of the loss is arbitrary relative to the prior. Popular approaches for tuning $\eta$ include cross-validation \citep{germain2009pac,thiemann2017strongly}, hyperpriors \citep{bissiri2016general,rigon2023generalized}, and matching credible intervals to confidence intervals \citep{syring2019calibrating}. A recent review by \cite{wu2023comparison} compares calibration methods from \cite{syring2019calibrating,lyddon2019general,holmes2017assigning,grunwald2017inconsistency}. To provide a Bayesian approach to inference, which is also acceptable to frequentists, it is appealing for $\eta$ to be chosen so that credible intervals from \eqref{eq:gibbs} have correct coverage \cite[Section 4]{martin2022direct}. Unfortunately, this is often impossible even in simple situations.

\textcolor{black}{Principal component analysis provides a key example where current technology falls short, and is one motivating application for our work.} Failure to account for uncertainty in components results in under-representation of uncertainty in downstream inference. Accounting for uncertainty is a difficult task: principal component analysis is routinely applied to complicated, high-dimensional datasets with relatively small sample sizes, resulting in violations of assumptions for loss-based uncertainty quantification and practical challenges in choosing a realistic likelihood. \textcolor{black}{Furthermore, orthonormality constraints imply that components are supported on a manifold, and naively constructed Euclidean intervals may be misleading.} Conceptually, these factors make principal component analysis an excellent use-case for Gibbs posteriors. However, in practice we find \eqref{eq:gibbs} cannot produce credible intervals for components with correct or near correct coverage. Similar coverage problems arise broadly when using Gibbs posteriors to study multiple quantities of interest.

We propose a generalization of Gibbs posteriors that overcomes these shortcomings by allowing a different tuning parameter controlling uncertainty for each quantity of interest. \textcolor{black}{Our framework is tailored to sequential problems, where each quantity is connected to the data only through a loss function and each loss naturally depends on previously estimated quantities. This encompasses many applications, including principal component analysis, multi-scale modeling, tensor factorizations, and general hierarchical models.} We extend existing Gibbs posterior theory, establishing concentration and a Bernstein-von Mises theorem under weak assumptions for losses defined on manifolds. Taking the loss to be a negative log-likelihood, we obtain what we believe is the first Bernstein-von Mises for arbitrary traditional likelihood-based Bayesian posteriors supported on manifolds. Our conditions can be verified with calculus in any chart, and do not require any advanced differential geometry machinery. \textcolor{black}{The theoretical and practical utility of our approach is made concrete through an application to principal component analysis.}

\section{Sequential Gibbs Posteriors}\label{sec:gibbs}

\subsection{Motivation}\label{sec:motivation}
We begin with a simple example illustrating the failure of Gibbs posteriors and motivating our proposed solution. Consider estimating the mean $\mu = E(X)\in\mathbb{R}$ by minimizing the risk
\begin{align*}
    R(\mu) &= \frac{1}{2}\int (x-\mu)^2P(dx).
\end{align*}
Since one does not know the true data generating measure $P$, it is standard to minimize the empirical risk based on independent and identically distributed samples $x=(x_1,\dots,x_n)$, 
\begin{align*}
    \ell(\mu\mid x) &=  \frac{1}{2n}\sum_{i=1}^n(x_i-\mu)^2.
\end{align*}
Inference for $\mu$ can be performed without assumptions about $P$ by defining a Gibbs posterior using the empirical risk \citep{martin2022direct}. Adopting a uniform prior, \eqref{eq:gibbs} becomes
\begin{align*}
    \pi_{\eta_\mu}(\mu\mid x) &\propto \exp\bigg\{-\frac{n\eta_\mu}{2n}\sum_{i=1}^n(x_i-\mu)^2\bigg\} \propto N\bigg(\mu;\frac{1}{n}\sum_{i=1}^nx_i, \frac{1}{n\eta_{\mu}} \bigg).
\end{align*}
The Gibbs posterior is a normal distribution centered at the sample mean and $n\eta_\mu$ is the posterior precision. Equal tailed credible intervals for $\mu$ will be centered at the sample mean and can be made larger or smaller by decreasing or increasing $\eta_\mu$. Now consider estimating the variance $\sigma^2=\text{var}(X)$ conditional on $\mu$ by minimizing 
\begin{align*}
    R(\sigma^2\mid \mu) = \frac{1}{2}\int \{\sigma^2-(x-\mu)^2\}^2P(dx).
\end{align*}

\begin{table}
\begin{tabular}{lcccc}
& N$(0,1)$ & $\text{t}_5(0,1)$ & S-N$(0,1,1)$ & Gumbel$(0,1)$ \\[5pt]
Joint Gibbs & 60 & 54 & 35 & 0 \\
Sequential Gibbs & 95 & 95 & 95 & 95 \\
\end{tabular}
\label{tablelabel}
\vspace*{1em}
\caption{\textit{Estimated coverage of $95\%$ credible intervals for $\mu$}. Coverage of credible intervals for $\mu$ after tuning $\eta$ so $95\%$ credible intervals for $\sigma^2$ had $95\%$ coverage. S-N denotes the Skew-Normal distribution.}
\label{table:coverage}
\end{table}

\noindent This risk is minimized by $E(X^2) +\mu^2 - 2E(X)\mu$, which is equal to the variance if $\mu=E(X)$. 
The Gibbs posterior defined by the empirical risk is
\begin{align*}
    \pi_{\eta_{\sigma^2}}(\sigma^2\mid x,\mu) &\propto \exp\bigg[-\frac{n\eta_{\sigma^2}}{2n}\sum_{i=1}^n\{\sigma^2-(x_i-\mu)^2\}^2\bigg] \propto N_{(0,\infty)}\bigg\{\sigma^2; \frac{1}{n}\sum_{i=1}^n(x_i-\mu)^2, \frac{1}{n\eta_{\sigma^2}}\bigg\}.
\end{align*}
If $\mu$ is the sample mean, then the mode of this distribution is the sample variance. As before, $\eta_{\sigma^2}$ acts as a precision parameter that can be used to control the width of credible intervals. These two Gibbs posteriors can be used separately for coherent Bayesian inference on the mean and variance, and can be tuned so credible intervals have correct coverage for a wide array of distributions. However, problems arise in performing joint inference on both parameters with a single Gibbs posterior. Inducing a joint posterior over $(\mu, \sigma^2)$ with \eqref{eq:gibbs} requires defining a combined loss, which fixes the scale of one parameter relative to the other, resulting in poor coverage for at least one parameter in many situations. For example, summing the two losses leads to the Gibbs posterior
\begin{align*}
    \pi_\eta(\mu, \sigma^2\mid x) &\propto \exp\bigg(-\frac{n\eta}{2n}\sum_{i=1}^n[(x_i-\mu)^2 + \{\sigma^2-(x_i-\mu)^2\}^2]\bigg),
\end{align*}
which is not a recognizable distribution, but can be sampled via Metropolis-Hastings. From \eqref{eq:gibbs}, $\eta$ controls dispersion for both $\mu$ and $\sigma^2$. \Cref{table:coverage} highlights the catastrophically poor coverage of credible intervals for $\mu$ after tuning $\eta$ so $95\%$ credible intervals for $\sigma^2$ have $95\%$ coverage. Details on tuning these posteriors are in the supplement.

Motivated by this shortcoming, we propose to avoid combining the risks into a single loss function and instead base inference on the unique joint distribution defined by conditional Gibbs posteriors for each loss:
\begin{align*}
    \pi_{\eta_\mu, \eta_{\sigma^2}}(\mu,\sigma^2\mid x) &= \pi_{\eta_\mu}(\mu\mid x)\pi_{\eta_{\sigma^2}}(\sigma^2\mid x,\mu) \\
    &=N\bigg(\mu;\frac{1}{n}\sum_{i=1}^nx_i, \frac{1}{n\eta_\mu} \bigg)N_{(0,\infty)}\bigg\{\sigma^2; \frac{1}{n}\sum_{i=1}^n(x_i-\mu)^2, \frac{1}{n\eta_{\sigma^2}}\bigg\}.
\end{align*}
Importantly, the hyperparameters $\eta_\mu$ and $\eta_{\sigma^2}$ can be tuned to ensure good coverage for both parameters across a variety of distributions for $x$ (\Cref{table:coverage}). In the next section we formalize our sequential Gibbs posterior construction.

\subsection{The Sequential Posterior}

Our goal is to perform inference on $J$ parameters $\theta_j\in\mathcal{M}_j$, $j\in [J]=1,\ldots,J$, connected to observed $\mathcal{X}$-valued data $x=(x_1,\ldots,x_n)\in\mathcal{X}^n$ by a sequence of real-valued loss functions,
\begin{align*}
 \ell_j^{(n)}:\mathcal{M}_j\times \mathcal{X}^n\times \mathcal{M}_{<j}\to\mathbb{R}, 
\end{align*}
where $\mathcal{X}$ is an arbitrary set, $\mathcal{M}_j$ is a manifold corresponding to the parameter space for $\theta_j$, and $\mathcal{M}_{<j}=\otimes_{k=1}^{j-1}\mathcal{M}_k$. All manifolds in this work are assumed smooth and orientable. Orientability ensures the existence of a volume form which serves as our default reference measure\footnote{In particular, all densities are implicitly with respect to the volume form.} and is equivalent to Lebesgue measure in the Euclidean setting. The $j$th loss measures congruence between $\theta_j$ and the data conditional on $\theta_{<j}=(\theta_1,\dots,\theta_{j-1})$. By allowing parameters restricted to manifolds, we encompass both unrestricted real-valued parameters and more complex settings, such as in principal component analysis when orthogonality constraints are included. This setup is broad and includes supervised and unsupervised loss functions. Our general results require neither independent, identically distributed data, nor assumptions of model correctness. \textcolor{black}{In an effort to make our theory broadly applicable, we do not place direct assumptions on how losses depend on the data $x$. This is more general than the typical Gibbs posterior framework which requires the loss to be additive as a function of the data \citep{bissiri2016general}.}

\begin{example}[Multi-scale inference]\label{ex:multiscale}
It is often useful to study data at different levels of granularity, such as decomposing a temperature distribution into a global component, a regional component, and local variation. Practical problems often occur when fitting these models jointly, as it is possible for the fine-scale component to explain the data arbitrarily well. To resolve this, a sequence of losses can be defined estimating first the coarse scale, then the medium scale conditional on the coarse scale, and so on \citep{fox2012multiresolution, nychka2015multiresolution, katzfuss2017multi, peruzzi2018bayesian}. For example, let $h_1>\cdots>h_J>0$ be a set of decreasing bandwidths and $f_j$ be mean zero Gaussian processes with kernels $K_j(x,x')=\exp\{-(x-x')/h_j\}$, $j\in[J]$. At the coarsest scale, we may model $y=f_1(x)+\varepsilon_1$ where $y$ is a response, $x$ is a feature, and $\varepsilon_1\sim N(0,\sigma_1^2)$ are errors. The negative log-likelihood defines a loss for $f_1$. Conditional on $f_1$, we model $y-f_1(x)=f_2(x)+\varepsilon_2$ with errors $\varepsilon_2\sim N(0,\sigma_2^2)$; again the negative log-likelihood defines a loss for $f_2$. Proceeding sequentially, we obtain losses for $f_{j}\mid f_1,...f_{j-1}$.  Similar decompositions occur broadly within spatial statistics, time series analysis, image analysis, tree-based models, and hierarchical clustering.
\end{example}

\begin{example}[Matrix/tensor factorization]\label{ex:matrix}
It is routine to decompose matrices and tensors as a sum of low-rank components, as in principal component analysis. These models are often fit by recursively finding and then subtracting the best rank $1$ approximation, defining a sequence of losses depending on previously estimated parameters \citep{wold1966estimation, lee1999learning, ng2001spectral, kolda2009tensor, hout2013multidimensional}. For example, let $X$ be a $k$-tensor of dimension $d_1\times\cdots\times d_k$ and consider fitting a rank $J$ approximation by iteratively finding and subtracting $J$ rank $1$ approximations. The best rank $1$ approximation minimizes 
\begin{align*}
    \ell_1^{(n)}(x^{(1)}\mid X) &= \lVert X - \lambda^{(1)} x_1^{(1)}\otimes\cdots\otimes x_k^{(1)}\rVert^2
\end{align*}
where $\lambda^{(1)}\in\mathbb{R}$, $x_i^{(1)}\in\mathbb{R}^{d_i}$, $i=1,...,k$, and $x^{(1)}=\{\lambda^{(1)}, x_1^{(1)},...,x_k^{(1)}\}$. Letting $\hat X_1 = \lambda_1 x_1^{(1)}\otimes \cdots\otimes x_k^{(1)}$ be the reconstructed tensor, the next best rank $1$ approximation minimizes
\begin{align*}
    \ell_2^{(n)}(x^{(2)}\mid X, x^{(1)}) &= ||X - \hat X_1- \lambda^{(2)} x_1^{(2)}\otimes \cdots\otimes x_k^{(2)}||^2,
\end{align*}
and so on. Characterizing uncertainty can be difficult in these settings due to high dimensionality and manifold constraints such as orthogonality. 
\end{example}

We now define the sequential posterior.
\begin{definition}[The sequential Gibbs posterior]\label{def:seq_gibbs}
    Given losses $\ell_j^{(n)}$, priors $\Pi^{(0)}_j$ on $\mathcal{M}_j$, and precision hyperparameters $\eta_j>0$, $j\in[J]$, the sequential Gibbs posterior is
\begin{align}
    \Pi_\eta^{(n)}(d\theta_1,\dots,d\theta_J\mid x) &= \prod_{j=1}^{J}\frac{1}{z_j^{(n)}(x,\theta_{<j})}\exp\{-\eta_jn\ell_j^{(n)}(\theta_j\mid x, \theta_{<j})\} \Pi^{(0)}_j(d\theta_j),\label{eq:seq} \\
    z_j^{(n)}(x,\theta_{<j}) &=\int_{\mathcal{M}_j} \exp\{-\eta_jn\ell_j^{(n)}(\theta_j\mid x, \theta_{<j})\} \Pi^{(0)}_j(d\theta_j). \nonumber
\end{align}
All results in this work assume $z^{(n)}_j(x,\theta_{<j})<\infty$ for every $\theta_{<j}\in\mathcal{M}_{<j}$. \textcolor{black}{This holds whenever $\ell^{(n)}_j$ is uniformly bounded from below---for example, when $\ell^{(n)}_j\geq 0$, as is common in practice---since then $\exp\{-\eta_jn\ell_j^{(n)}(\theta_j\mid x, \theta_{<j})\}$ is bounded above, and hence $z^{(n)}_j(x,\theta_{<j})$ is finite.}
\end{definition}
\cite{bissiri2016general} consider all coherent generalizations of Bayes' rule for updating a prior based on \textcolor{black}{a data-additive loss} and derive \eqref{eq:gibbs} as the unique optimal decision-theoretic update. Our sequential Gibbs posterior is the unique joint distribution with Gibbs posteriors for each conditional $\theta_j\mid x,\theta_{<j}$, and therefore trivially retains the coherence, uniqueness, and optimality properties of \cite{bissiri2016general} for \textcolor{black}{data-additive losses}. 
\textcolor{black}{The sequential Gibbs posterior is not equivalent to using \eqref{eq:gibbs} with combined loss $\eta_1\ell_1^{(n)}+\cdots + \eta_J\ell_J^{(n)}$, as the normalizing constants have considerable influence on the joint distribution.}

\subsection{Large Sample Asymptotics}

We now study frequentist asymptotic properties of the Gibbs posterior \eqref{eq:gibbs} and sequential Gibbs posterior \eqref{eq:seq}. Current theory for \eqref{eq:gibbs} with Euclidean parameters provides sufficient conditions under which the posterior has a limiting Gaussian distribution \citep{miller2021asymptotic, martin2022direct}. Theorem \ref{thrm:bvm} extends these results, providing sufficient conditions for \eqref{eq:gibbs} to converge to a Gaussian distribution as $n\to\infty$ when parameters are supported on manifolds. Taking the loss as a negative log-likelihood, this provides new asymptotic theory for traditional Bayesian posteriors on non-Euclidean manifolds. Theorems \ref{thrm:concentration} and \ref{thrm:sequential} concern the sequential Gibbs posterior, \eqref{eq:seq}. \Cref{thrm:concentration} establishes concentration over general metric spaces, and \Cref{thrm:sequential} extends \Cref{thrm:bvm} to the sequential setting, providing sufficient conditions for \eqref{eq:seq} to converge to a Gaussian distribution as $n\to\infty$. Formalizing these notions requires several assumptions on the losses, their minima, and their limits. All proofs are in \Cref{sec:proofs} of the online supplement. Additional assumptions are also deferred to \Cref{sec:proofs} to minimize notation in the main text; all are manifold and/or sequential analogues of standard assumptions for Euclidean, non-sequential Gibbs posteriors \citep{miller2021asymptotic} and are often  simple to verify in practice.

In the following, \textcolor{black}{a chart $(U, \varphi)$ on a $p$-dimensional manifold $\mathcal{M}$ is an open $U\subseteq\mathcal{M}$ and a diffeomorphism $\varphi:U\to\varphi(U)\subseteq\mathbb{R}^p$. The support of a measure $\mu$ on $\mathcal{M}$ is $\mathrm{supp}(\mu)=\{\theta\in\mathcal{M} : \mu(U)>0\ \text{for all open neighborhoods}\ U\ \text{of}\ \theta\}$. The pushforward of $\mu$ by a measurable function $f:\mathcal{M}\to f(\mathcal{M})$ is the measure $f_\#\mu$ on $f(\mathcal{M})$ defined by $f_\# \mu(A)=\mu\{f^{-1}(A)\}$ for all measurable $A\subseteq f(\mathcal{M})$.} The total variation between measures $P$ and $Q$ is denoted $d_{TV}(P,Q)$. \textcolor{black}{Finally, $f'$ and $f''$ denote the first and second derivatives of $f$, respectively. A detailed description of derivatives on manifolds is provided in \Cref{sec:geometry} of the supplement.}

\begin{theorem}\label{thrm:bvm}
Let $\mathcal{M}$ be a manifold, let $\ell^{(n)}:\mathcal{M}\times\mathcal{X}^n\to\mathbb{R}$ be a sequence of functions converging almost surely to $\ell:\mathcal{M}\to\mathbb{R}$, and let $\Pi_\eta^{(n)}$ be the Gibbs posterior in \eqref{eq:gibbs}. Suppose Assumptions \ref{as:third} and \ref{as:hessian} hold and that $(U,\varphi)$ is any chart on $\mathcal{M}$ satisfying
\begin{enumerate}[label=(\alph*)]
\item There exists a $\phi^\star$ in the interior of a compact $K\subseteq U$ such that $\ell(\theta) > \ell(\phi^\star)$ for all $\theta\in U \cap (K\setminus\{\phi^\star\})$ and $\liminf_n\inf_{\theta\in U\setminus K}\{\ell^{(n)}(\theta)- \ell(\phi^\star)\} > 0$.

\item $\Pi^{(0)}$ has a density $\pi^{(0)}$ that is continuous and strictly positive at $\phi^\star$, and $\mathrm{supp}(\Pi^{(0)}) \subseteq U$.
\end{enumerate}
If there exists a sequence $\theta^{(n)}\to\phi^\star$ such that $(\ell^{(n)})'(\theta^{(n)}\mid x)=0$, then $\ell'(\phi^\star)=0$ and
\begin{align*}
	d_{TV}\{(\tau^{(n)}\circ \varphi)_\#\Pi^{(n)}_\eta, N(0, \eta^{-1}H_{\varphi}^{-1})\}\to 0
\end{align*}
almost surely, where $\tau^{(n)}(\tilde\theta)= \sqrt{n}\{\tilde\theta-\varphi(\theta^{(n)})\}$ and \textcolor{black}{$H_{\varphi} = (\ell\circ\varphi^{-1})''\{\varphi(\phi^\star)\}$}. 
\end{theorem}
In \Cref{thrm:bvm}, $\theta^{(n)}$ minimizes the finite sample loss and is mapped to Euclidean space via $\varphi$ to obtain $\tilde\theta^{(n)}=\varphi(\theta^{(n)})$. Samples from $\theta\sim\pi_\eta^{(n)}$ are mapped to Euclidean space to produce $\tilde\theta=\varphi(\theta)$, centered by subtracting $\tilde\theta^{(n)}$, then scaled by $\sqrt{n}$. Asymptotically, this results in samples from a Gaussian distribution $\sqrt{n}(\tilde\theta-\tilde\theta^{(n)})\approx N(0,\eta^{-1}H_\varphi^{-1})$. The total variation distance between these distributions vanishes almost surely \citep{miller2021asymptotic}, which is stronger than the usual guarantees in probability \citep{martin2022direct}. The covariance of the limiting Gaussian is $\eta^{-1}H_\varphi^{-1}$, where $H_\varphi$ is the Hessian of $\ell\circ\varphi^{-1}$ evaluated at $\varphi(\phi^\star)$.

\textcolor{black}{Assumptions (a) and (b) are the only chart-dependent conditions in \Cref{thrm:bvm}. Assumption (a) introduces the local minimizer $\phi^\star$ of $\ell$; since it is in the interior of $K$, \Cref{thrm:bvm} applies to manifolds with or without boundary. Assumption (b) implies $\mathrm{supp}(\Pi^{(n)}_\eta)\subseteq U$, and therefore $(\tau^{(n)}\circ\varphi)_{\#}\Pi^{(n)}_\eta$ is a valid probability distribution for all $n$. Assumption (b) is less restrictive than it may appear because, by \Cref{thrm:concentration} when $J=1$ or, equivalently, by \citet[Theorem 3]{miller2021asymptotic}, the posteriors $\Pi^{(n)}_\eta$ concentrate around $\phi^\star$ asymptotically for any prior whose support contains $\phi^\star$. In \Cref{sec:bvm_proof} we show that a result similar to \Cref{thrm:bvm} holds without restrictions on the support of the prior. Specifically, under the assumptions of \Cref{thrm:bvm}, but with $\mathrm{supp}(\Pi^{(0)})\subseteq U$ replaced by $\varphi(U)=\mathbb{R}^p$, we prove that
\begin{align}\label{eq:bvm_alternative}
    d_{TV}\{\Pi^{(n)}, (\tau^{(n)}\circ\varphi)^{-1}_\#N(0,H_\varphi^{-1})\} &\to 0
\end{align}
almost surely. This result applies to a wide range of charts: For example, every open, convex subset of $\mathbb{R}^p$ is diffeomorphic to $\mathbb{R}^p$. Thus, if $(U,\varphi)$ is a chart with $\varphi(U)$ a (necessarily open) convex subset of $\mathbb{R}^p$, we can compose $\varphi$ with a diffeomorphism $\psi$ such that $\psi\{\varphi(U)\}=\mathbb{R}^p$. The resulting chart, $(U,\psi\circ\varphi)$, then satisfies the requisite condition.} Assumptions \ref{as:third} and \ref{as:hessian} are standard and amount to control over third derivatives and positive-definiteness of $\ell''(\phi^\star)$, respectively. \Cref{lem:invariant} guarantees both assumptions can be verified using basic calculus in any chart containing $\phi^\star$; no additional differential geometry is required.

We emphasize that \Cref{thrm:bvm} applies to any density that can be written as a Gibbs posterior, including likelihood-based posteriors. Posteriors over manifolds arise in a diverse array of applications, including covariance modelling (positive semidefinite matrices), linear dimensionality reduction (Grassmann manifold), directional statistics (spheres and Stiefel manifolds), and shape analysis (Kendall's shape space) \citep{schmidler2007fast, holbrook2016bayesian, elvira2017bayesian, hernandez2017general, lin2017bayesian, pourzanjani2021bayesian, thomas2022learning}. Despite this interest, to our knowledge, there is no Bernstein-von Mises theorem on manifolds, even for parameters on spheres. Existing asymptotic literature focuses on specific estimates such as the Frechet mean or M-estimators, and provides much weaker guarantees than total variation convergence of the entire posterior to a Gaussian \citep{kendall2011limit, bhattacharya2014statistics, eltzner2019smeary, paindaveine2020inference, ciobotaru2022consistency}. We believe \Cref{thrm:bvm} is the first result providing intuition and frequentist justification for the limiting behaviour of this broad class of Bayesian models.

\textcolor{black}{There is nothing sequential about \Cref{thrm:bvm}, which is precisely \Cref{thrm:sequential} when $J=1$. The cases $J=1$ and $J>1$ are considered separately because \Cref{thrm:bvm} has simpler assumptions and, as detailed above, is of broad interest.}
\textcolor{black}{Theorems \ref{thrm:concentration} and \ref{thrm:sequential} are our sequential results. These require the following assumptions, which represent sequential extensions of certain assumptions in \Cref{thrm:bvm}.} 

\begin{assumption}\label{as:converge}
For all $j\in[J]$ there exist $\ell_j:\mathcal{M}_j\times \mathcal{M}_{<j}\to\mathbb{R}$ such that $\ell_j^{(n)}(\cdot\mid x,\theta_{<j})\to \ell_j(\cdot\mid \theta_{<j})$ almost surely for every $\theta_{<j}$.
\end{assumption}

\begin{assumption}\label{as:min}
For all $j\in[J]$ there exist $\theta_j^{(n)}:\mathcal{M}_{<j}\to \mathcal{M}_j$ and $\theta_j^\star:\mathcal{M}_{<j}\to\mathcal{M}_j$ satisfying $\ell_j^{(n)'}\{\theta_j^{(n)}(\theta_{<j})\mid x,\theta_{<j}\}=0$ almost surely and $\ell_j'\{\theta_j^\star(\theta_{<j})\mid\theta_{<j})=0$.
\end{assumption}

\begin{assumption}\label{as:seq_min}
\textcolor{black}{Let $\theta_j^\star$ be as in \Cref{as:min} and define $\phi_j^\star$ iteratively as follows: Set $\phi_1^\star=\theta_1^\star$, and, for $j>1$, set $\phi_j^\star=\theta_j^\star(\phi_{<j}^\star)$, where $\phi_{<j}^\star=(\phi_1^\star,\dots,\phi_{j-1}^\star)$. Assume each $\phi_j^\star$ is in the interior of a compact $K_j\subseteq\mathcal{M}_j$ and satisfies $\ell_j(\theta_j\mid\phi_{<j}^\star) > \ell_j(\phi_j^\star\mid\phi_{<j}^\star)$ for all $\theta_j\in K_j\setminus\{\phi_j^\star\}$.}
\end{assumption}

\Cref{as:converge} guarantees losses have non-degenerate limits and is satisfied, for example, if the losses are empirical risk functions, as the strong law of large numbers guarantees almost sure convergence to the true risk function. Assumptions \ref{as:min} and \ref{as:seq_min} introduce optimizers of the conditional losses; these are naturally functions of previously estimated parameters. \textcolor{black}{The $\phi_j^\star$ in \Cref{as:seq_min} are sequential minimizers of the limiting losses. Specifically, $\phi_1^\star$ is the unique minimizer of $\ell_1(\cdot)$ in $K_1$, $\phi_2^\star$ is the unique minimizer of $\ell_2(\cdot\mid\phi_1^\star)$ in $K_2$, 
and so on. As with \Cref{thrm:bvm}, the assumption that each $\phi_j^\star$ is in the interior of $K_j$, and hence of $\mathcal{M}_j$, implies that \Cref{thrm:sequential} holds for manifolds with or without boundary. The sequential minimizers can differ substantially from the values obtained by jointly minimizing the total loss. Our framework assumes that the statistical problem is fundamentally sequential.}

Fix metrics $d_j$ on $\mathcal{M}_j$. While \Cref{thrm:concentration} holds for various metrics on $\mathcal{M}$, including $\max\{d_1,...,d_J\}$ and $(d_1^p+\cdots+d_J^p)^{1/p}$ for $p\geq1$, we focus on $d = (d_1^2 + \cdots + d_J^2)^{1/2}$.

\begin{theorem}\label{thrm:concentration}
Suppose Assumptions \ref{as:converge}, \ref{as:min}, \ref{as:seq_min}, and \ref{as:concentration} hold and that $\Pi^{(0)}_j(N_{j,\epsilon})>0$ for all $\epsilon>0$, where $N_{j,\epsilon}=\{\theta_j:d_j(\theta_j,\phi^\star_j)<\epsilon\}$. Set $\phi^\star = (\phi_1^\star,\dots,\phi_J^\star)$ and $N_\epsilon=\{\theta:d(\theta,\phi^\star)<\epsilon\}$, with $d = (d_1^2 + \cdots + d_J^2)^{1/2}$ as above. Then $\Pi^{(n)}_\eta(N_\epsilon)\to 1$ almost surely for all $\eta$ and $\epsilon>0$.
\end{theorem}

\Cref{thrm:concentration} ensures samples from the sequential Gibbs posterior concentrate around the point $\phi^\star$. The proof relies on additional regularity conditions, namely continuity (parts (a) and (b) of \Cref{as:concentration}) and well-separated minimizers (part (c) of \Cref{as:concentration}).

We now present the sequential analogue of \Cref{thrm:bvm}.

\begin{theorem}\label{thrm:sequential}
Suppose Assumptions \ref{as:converge}, \ref{as:min}, \ref{as:seq_min}, \ref{as:third_seq}, and \ref{as:hess_seq} hold, and let $\Pi_{\eta}^{(n)}$ be the sequential Gibbs posterior in \eqref{eq:seq}. For all $j\in[J]$, assume that $(U_j,\varphi_j)$ is a chart on $\mathcal{M}_j$ containing $\phi_j^\star$ and satisfying Assumptions \ref{as:min_sep_seq}-\ref{as:loss_seq}. If each \textcolor{black}{$\Pi^{(0)}_j$ has a density $\pi^{(0)}_j$ that is continuous and strictly positive at $\phi_j^\star$, and $\mathrm{supp}(\Pi_j^{(0)}) \subseteq U_j$}, then
\begin{align*}
    (\tau^{(n)}\circ \varphi)_\# \Pi^{(n)}_\eta\to \prod_{j=1}^J N(0,\eta_j^{-1}H_{\varphi_j}^{-1})
\end{align*}
setwise, where \textcolor{black}{$H_{\varphi_j}=\{\ell_j(\cdot\mid \phi_{<j}^\star)\circ\varphi_j^{-1}\}''\{\varphi_j(\phi_j^\star)\}$}. Here $\varphi=(\varphi_1,\dots,\varphi_J):\otimes_{j=1}^J U_j\to\otimes_{j=1}^J\mathbb{R}^{p_j}$ and, setting $\tilde \theta_j^{(n)}(\theta_{<j}) = \varphi_{j}\{\theta_j^{(n)}(\theta_{<j})\}$, $\tau^{(n)}:\otimes_{j=1}^J\mathbb{R}^{p_j}\to \otimes_{j=1}^J\mathbb{R}^{p_j}$ is defined by
    \begin{align*}
        \tau^{(n)}(\tilde \theta) = \sqrt{n}\left\{\tilde\theta_1-\tilde\theta_1^{(n)},\tilde\theta_2-\tilde\theta_2^{(n)}(\theta_1),\dots,\tilde\theta_J-\tilde\theta^{(n)}_J(\theta_{<J})\right\}.
    \end{align*}
\end{theorem}
When $J=2$, one samples $(\theta_1,\theta_2)$ by drawing $\theta_1\sim \pi_{\eta_1}^{(n)}$ and $\theta_2\mid\theta_1\sim \pi_{\eta_2}^{(n)}(\cdot\mid\theta_1)$. These are mapped to Euclidean space to obtain $\tilde \theta_1=\varphi_1(\theta_1)$ and $\tilde \theta_2=\varphi_2(\theta_2)$. The finite sample minimizers $\theta_1^{(n)}$ and $\theta_2^{(n)}(\theta_1)$ are computed and mapped to Euclidean space to obtain $\tilde\theta_1^{(n)}=\varphi_1(\theta_1^{(n)})$ and $\tilde\theta_2^{(n)}(\theta_1)=\varphi_2\{\theta_2^{(n)}(\theta_1)\}$. Centering and scaling gives $\sqrt{n}(\tilde\theta_1-\tilde\theta_1^{(n)})\approx N(0,\eta_1^{-1}H_{\varphi_1}^{-1})$ and $\sqrt{n}\{\tilde\theta_2-\tilde\theta_2^{(n)}(\theta_1)\}\approx N(0,\eta_2^{-1}H_{\varphi_2}^{-1})$. \textcolor{black}{Asymptotically, $\theta_1$ and $\theta_2$ are independent; intuitively this happens because $\theta_1$ concentrates at $\theta_1^\star$, so for large $n$ we have $\theta_2\mid\theta_1\approx \theta_2\mid\theta_1^\star$.} As in \Cref{thrm:bvm}, the limiting covariances are inverse Hessians of the losses evaluated at critical points, \textcolor{black}{and the assumption that $\mathrm{supp}(\Pi_j^{(0)}) \subseteq U_j$ can be replaced with $\varphi_j(U_j)=\mathbb{R}^{p_j}$, yielding
\begin{align*}
    \Pi^{(n)}_\eta\to \prod_{j=1}^J (\tau_j^{(n)}\circ \varphi_j)^{-1}_\# N(0,\eta_j^{-1}H_{\varphi_j}^{-1})
\end{align*}
setwise as $n\to\infty$.} Assumptions \ref{as:third_seq}-\ref{as:loss_seq} are natural extensions of those in \Cref{thrm:bvm}, with only Assumptions \ref{as:min_sep_seq}-\ref{as:loss_seq} depending on the specified charts; see \Cref{sec:proofs}. \textcolor{black}{The assumptions serve to guarantee uniform convergence of the relevant functions and minimizers.}

Theorems \ref{thrm:bvm} and \ref{thrm:sequential} highlight the role of $\eta$ as a precision parameter. The sequential Gibbs posterior has individual tuning parameters for each $\theta_j$ and therefore has greater flexibility. In the following subsection, we develop a practical algorithm to take advantage of this flexibility to tune the sequential posterior so that credible intervals for $\theta_j$ are approximately valid confidence intervals. \textcolor{black}{Perfect asymptotic coverage for all combinations of parameters is only possible when the limiting covariance has a sandwich form, which does not hold for general loss functions \citep{martin2022direct}. Instead, we focus on marginal coverage, which provides a reasonable notion of uncertainty for problems where existing probabilistic methods fall short.}

\subsection{Calibration}
We propose a bootstrap-based calibration algorithm for tuning the sequential Gibbs posterior so credible intervals have approximately valid frequentist coverage, without reliance on asymptotic results or strong parametric assumptions. \textcolor{black}{Our algorithm is suitable for problems where existing calibration methods are computationally infeasible or not applicable.} We are inspired by the general posterior calibration algorithm in \cite{syring2019calibrating}, which uses Monte Carlo within the bootstrap to estimate coverage of credible regions and iteratively updates $\eta$ to drive coverage to a desired value. Sampling the posterior over each bootstrap replicate at each iteration of the algorithm is computationally intensive, rendering this approach impractical for principal component analysis in moderate-to-high dimensions. In the sequential setting, the computational burden is compounded by the need to calibrate $J$ different hyperparameters. Motivated by this, we propose a new general calibration algorithm which matches the volume of credible regions to the volume of pre-calculated bootstrap confidence regions. Pre-calculating the volume of a confidence region avoids the need to sample within the bootstrap and dramatically reduces the computational burden of calibration. Calculating volumes on manifolds can be difficult; we avoid this by restricting credible/confidence regions to be \textcolor{black}{geodesic balls}, which reduces matching volumes to matching radii. 

We now outline the procedure for a Gibbs posterior with a single loss, dropping redundant subscripts for readability. Fix a distance $d$ on $\mathcal{M}=\mathcal{M}_1$ and let $N_r(\xi)=\{\theta\in\mathcal{M}\mid d(\theta,\xi)<r\}$ be the ball of radius $r$ around $\xi\in\mathcal{M}$. Let $\hat \phi(x)$ be the minimizer of $\ell^{(n)}(\cdot\mid x)$; we use this as a finite sample estimator of $\phi^\star$. The frequentist coverage of the ball $N_r\{\phi(x)\}$ is
\begin{align*}
    c(r)&= E_{x\sim P_x}(1[\phi^\star \in N_r\{\phi(x)\}])
\end{align*}
where $P_x$ is the sampling distribution of $n$ data points. Fix $\alpha\in(0,1)$. The radius $r^\star$ of a $100(1-\alpha)\%$ confidence ball satisfies
\begin{align*}
    r^\star &= \inf \{r>0\mid c(r) \geq 1-\alpha\}.
\end{align*}
We propose to choose $\eta$ so the Gibbs posterior assigns $100(1-\alpha)\%$ of its mass to $N_{r^\star}\{\phi(x)\}$, which would imply that the credible interval $N_{r^\star}\{\phi(x)\}$ has valid frequentist coverage. The probability mass the Gibbs posterior assigns to the confidence ball is
\begin{align*}
    m(\eta)&=E_{\theta\sim\pi_\eta} (1[\theta \in N_{r^\star}\{\phi(x)\}]),
\end{align*}
so calibrating the Gibbs posterior is equivalent to solving $m(\eta)=1-\alpha$. A solution $\eta_\alpha$ exists if $\Pi^{(0)}[N_{r^\star}\{\phi(x)\}] < 1-\alpha < 1$: this follows from the limits
\begin{align*}
    \lim_{\eta\to 0^+}m(\eta) = \Pi^{(0)}[N_{r^\star}\{\phi(x)\}], \quad \lim_{\eta\to \infty}m(\eta) = 1
\end{align*}
and the intermediate value theorem. One can calculate $\eta_\alpha$ with any suitable root finding method. In our experiments we estimate $m(\eta)$ via Monte Carlo and then use stochastic approximation \citep{robbins1951stochastic} to find $\eta_\alpha$. \textcolor{black}{Additional details, including the full algorithm, are in the online supplement.}

In practice we do not know $\phi^\star$ or the sampling distribution $P_x$, so $r^\star$ is unavailable. We overcome this by estimating the coverage function via the bootstrap,
\begin{align*}
    c(r) &\approx \frac{1}{B}\sum_{b=1}^B1[\phi(x) \in N_{r}\{\phi(x_b)\}],
\end{align*}
and then solving for $r^\star$ using this approximation. Here $x_b$ is a bootstrap replicate of $x$ and $B>0$ is an integer. Euclidean bootstrap confidence regions are known to have asymptotically correct coverage up to error terms of $O_p(1/n)$ under weak conditions \citep{hall2013bootstrap}, but these results are difficult to generalize to the case of balls on manifolds. In simulations we find this approximation produces well-calibrated Gibbs posteriors.

We calibrate the sequential posterior by applying the above procedure sequentially. Let $\hat \phi_j(x)\in\mathcal{M}_j$ be the point obtained by sequentially minimizing $\ell_1^{(n)}(\cdot\mid x),...,\ell_{j}^{(n)}(\cdot\mid x,\theta_{<j})$. The bootstrap is used to estimate the radii $\hat r_{j}$ of $100(1-\alpha)\%$ credible balls around $\hat \phi_j(x)$, $j\in[J]$. We tune $\eta_1$ so $\theta_1$ lies inside $N_{\hat r_1}\{\hat \phi_1(x)\}$ with probability $1-\alpha$; this parameter is then fixed and $\eta_2$ is tuned so $\theta_2$ lies inside $N_{\hat r_2}\{\hat \phi_2(x)\}$ with probability $1-\alpha$, and so on. In the next section we synthesize the above work on sequential posteriors, including asymptotic theory and finite sample tuning, to obtain a generalized posterior for principal component analysis.

\section{Application to Principal Component Analysis}\label{sec:pca}
\subsection{The Sequential Bingham Distribution}
Our sequential and manifold extensions to Gibbs posteriors are of broad interest, but were concretely motivated by principal component analysis. Recall principal component analysis projects high dimensional features $x_i\in\mathbb{R}^p$ to low dimensional scores $z_i\in\mathbb{R}^J$, $J<p$, contained in a plane $\mathcal{P}$. This defines $J$ new features, called components, which are linear combinations of the original $p$ features. Failure to characterize uncertainty in components and scores under represents uncertainty in downstream analysis.

Let $X\in\mathbb{R}^{n\times p}$ be a matrix of $n$ samples, centered so $x_1+\cdots+x_n=0$. The optimal plane $\hat{\mathcal{P}}$ minimizes the squared reconstruction error $||X-\mathcal{P}(X)||^2$, where $\mathcal{P}(X)$ is the projection of $X$ onto $\mathcal{P}$. It is well known that the leading unit eigenvectors $\{v^{(n)}_j\}_{j=1}^J$ of the empirical covariance $\hat\Sigma=X^TX/n$ form an orthonormal basis for $\hat{\mathcal{P}}$. These can be found by sequentially solving
\begin{align}
    v_j^{(n)}(v_{<j}) = \argmax_{v_j\in\mathbb{S}^{p-1}\cap\text{Null}\{v_1,...,v_{j-1}\}}v_j^T\hat\Sigma v_j,\quad j\in[J], \label{eq:PCA_loss1}
\end{align}
where $\text{Null}\{v_1,...,v_{j-1}\}$ is the null space of the span of $\{v_1,\dots,v_{j-1}\}$. The null space condition ensures eigenvectors are orthogonal; hence the matrix $\hat V\in\mathbb{R}^{p\times J}$ containing solutions of \eqref{eq:PCA_loss1} as columns is an element of the Stiefel manifold $\mathcal{V}(J,p)=\{V\in\mathbb{R}^{p\times J}\mid V^TV=I\}$. Computing charts on $\mathcal{V}(J,p)$ and sampling densities over $\mathcal{V}(J,p)$ is difficult. We instead use an equivalent formulation defined over spheres,
\begin{align}
    w_j^{(n)}(v_{<j}) = \argmax_{w_j\in\mathbb{S}^{p-j}} w_j^TN_{<j}^T\hat\Sigma N_{<j}w_j,\quad j\in[J], \label{eq:PCA_loss2}
\end{align}
where  $N_{<j}\in\mathbb{R}^{p\times p-j+1}$ is an orthonormal basis for $\text{Null}\{v_1,...,v_{j-1}\}$. The optimizer $w_j^{(n)}(v_{<j})$ is the leading eigenvector of $N_{<j}^T\hat\Sigma N_{<j}$ and is related to \eqref{eq:PCA_loss1} by $v_j^{(n)}(v_{<j})=N_{<j}w_j^{(n)}(v_{<j})$. 

Assuming uniform priors, our sequential posterior is 
\begin{align}
    \kappa_\eta^{(n)}(w\mid  x) &= \prod_{j=1}^J\frac{1}{z_{j}^{(n)}(w_{<j}\mid  x)}\exp(\eta_jn w_j^TN_{<j}^T\hat\Sigma N_{<j}w_j)\label{eq:pca_post}, \\
    \mbox{with}\quad z_{j}^{(n)}(w_{<j}\mid  x)&=\prescript{}{1}{F}_1\{1/2,(p-j)/2,\eta_jn N_{<j}^T\hat\Sigma N_{<j}\}, \nonumber
\end{align}
where $\prescript{}{1}{F}_1$ is the confluent hypergeometric function of matrix argument. This posterior is a product of Bingham distributions with concentration matrices $n\eta_jN_{<j}^T\hat\Sigma N_{<j}$. In \eqref{eq:pca_post}, $N_{<j}$ is computed using the samples $v_1,...,v_{j-1}$ which are found sequentially via the relations $v_j = N_{<j}w_j$. We write $\iota:\otimes_{j=1}^J\mathbb{S}^{p-j}\to\mathcal{V}(k,p)$ for the corresponding embedding $[w_1,...,w_J]\mapsto [v_1,...,v_J]$. The sequential Bingham distribution \eqref{eq:pca_post} can be used to sample posterior eigenvectors, providing a full characterization of uncertainty in components, scores, and any downstream inference involving these quantities. This can be done in isolation or jointly within a larger Bayesian model, which we illustrate in \Cref{sec:pcr}. 

Theorems \ref{thrm:concentration} and \ref{thrm:sequential} apply to \eqref{eq:pca_post}. To simplify presentation, we assume the data are centered with full-rank diagonal covariance; in this case the true components are $v_j^\star=e_j$, $j\in[J]$, where $e_j$ is the $j$th standard basis vector in $\mathbb{R}^p$. Samples from the sequential Gibbs posterior concentrate around the true eigenvectors, and centered/scaled samples converge to a Gaussian distribution with covariance proportional to the inverse eigengaps. A small technical detail is that \eqref{eq:pca_post} is antipodally symmetric, assigning equal mass to $\pm B$ for any measurable $B\subseteq\otimes_{j=1}^J\mathbb{S}^{p-j}$. We resolve this ambiguity by implicitly restricting the priors so $w\sim \pi_j$ implies $w_1>0$ almost surely.

\begin{proposition}\label{prop:pca_bvm}
    Assume $E(x)=0$ and $\text{var}(x)=\text{diag}(\lambda_1,...,\lambda_p)$ with $\lambda_1>\cdots>\lambda_p>0$. Fix charts $(U_j,\varphi_j)$ on $\mathbb{S}^{p-j}$ with $(1,0,...,0)\in U_j$, $j\in[J]$. Then $\iota(W)\to I_{p\times k}$ in probability where $W\sim\kappa_\eta^{(n)}$ and 
    \begin{align*}
    (\tau^{(n)}\circ\varphi)_\#\kappa_\eta^{(n)}\to \prod_{j=1}^JN\Big\{0,(2\eta_j)^{-1}H_{j}^{-1}\Big\}
    \end{align*}
    setwise, where $H_j^{-1}=\text{diag}\{(\lambda_j-\lambda_{j+1})^{-1},\ldots,(\lambda_j-\lambda_{p})^{-1}\}$ and $\tau^{(n)}$, $\varphi$ are as in \Cref{thrm:sequential}.
\end{proposition}

\Cref{prop:pca_bvm} views $\kappa_\eta^{(n)}$ as a density on $\otimes_{j=1}^J\mathbb{S}^{p-j}$. The charts $U_j\subseteq \mathbb{S}^{p-j}$ can be embedded in $\mathbb{S}^{p-1}$ via $w\to N_{<j}w$. For example, if $v_k=e_k$ for $k=1,...,j-1$, then $N_{<j}=[e_j,...,e_{J}]$ and $N_{<j}(1,0,...,0)^T=e_j$; hence the assumption $(1,0,...,0)\in U_j$ ensures that the $j$th eigenvector of $\text{var}(x)$ is in $N_{<j}U_{j}$. Any charts $(U_j,\varphi_j)$ with $(1,0,...,0)\in U_j$ can be used for \Cref{prop:pca_bvm}, such as $U_j=\{w\in\mathbb{S}^{p-j}\mid w_1>0\}$ and $\varphi_j(w)=w_{-1}$ where $w_{-1}\in\mathbb{R}^{p-j}$ is $w$ with the first entry removed. Other viable charts include the Riemannian logarithm, stereographic projection, or projective coordinates. See \Cref{sec:proofs} in the online supplement for the proof.

\begin{algorithm}[t]
\caption{Sampling from the sequential Bingham distribution}\label{alg:sample}
\KwData{Empirical covariance $\hat\Sigma=X^TX/n$ and positive hyperparameters $\eta_1,...,\eta_n$.}
\KwResult{$[v_1,...,v_j]\sim \kappa_\eta^{(n)}$}
\For{$j =1,...,J$}{
    $N_{<j}\gets \text{Null}\{v_1,...,v_{j-1}\}$\;
    $w_j\sim \text{Bing}(n\eta_jN_{<j}^T\hat\Sigma N_{<j})$; \\
    $v_j\gets N_{<j}w_j$; \\
}
\end{algorithm}

\subsection{Posterior Computation}

\textcolor{black}{Sampling from the sequential Gibbs posterior is straightforward if it is possible to generate exact samples from each of the conditional distributions, $\theta_j\mid \theta_{<j}$. Algorithm \ref{alg:sample} outlines this procedure for principal component analysis, where exact sampling is possible because the conditional Gibbs posteriors are Bingham distributions \citep{hoff2009simulation, kent2013new}.} Any algorithm which produces exact samples from a Bingham distribution, such as rejection sampling, can be combined with Algorithm \ref{alg:sample} to produce exact samples from \eqref{eq:pca_post}. Priors of the form $\pi^{(0)}_j(w_j)\propto \exp(Aw_j+b)$ can be accommodated by replacing the Bingham sampling step with a Fisher-Bingham sampling step \citep{hoff2009simulation}. The main bottleneck is computing $N_{<j}$. In high dimensions, $N_{<j}$ can be computed approximately \citep{nakatsukasa2022fast}, resulting in nearly orthogonal samples. 

\textcolor{black}{In general, sampling from the sequential Gibbs posterior may be complex due to the presence of $\theta_j$ in all normalizing constants $z_k^{(n)}(x,\theta_{<k})$, $k\geq j$. These constants are often intractable and do not cancel when calculating Metropolis-Hastings acceptance ratios. Ignoring the normalizing constants in calculating the acceptance ratio results in a Markov chain targeting a weighted version of the sequential posterior. This challenge also occurs with cut posteriors used in modular Bayesian analysis \citep{plummer2015cuts}. Multiple solutions have been developed for cut posteriors that are applicable for sampling sequential posteriors, including running intermediate chains until convergence, adjusting for bias during sampling using path-augmented proposal distributions, and adjusting for bias in functionals of interest after sampling using iterated expectations and coupled Markov chains \citep{jacob2020unbiased}. We explain these solutions in detail in Section S3.2 of the supplement, and provide practical guidance for problems where exact sampling is not feasible. Sampling over general manifolds introduces additional complexity typically requiring computationally intensive calculations of charts, geodesic maps, and associated Jacobians \citep{girolami2011riemann, betancourt2013general, wang2013adaptive}.}

\subsection{Simulations}
Uncertainty in eigenvectors depends on the marginal distributions of $X$ and $p/n$. We sample the rows of $X$ independently from a mean zero multivariate Gaussian or mean zero multivariate $t_5$ for each of the relative dimensions $p/n\in\{1/4,1/2,1\}$. All simulations fix $n=100$ and use a diagonal covariance where the first $k=5$ eigenvectors explain $90\%$ of the variance in the data. The first five eigenvalues are $(\lambda_1,\lambda_2,...,\lambda_5)=(10,9,...,6)$ and the remaining $p-5$ eigenvalues are linearly spaced and scaled to explain the remaining $10\%$ of the variance. We evaluate the coverage of multiple methods for estimating the first five eigenvectors. All credible/confidence balls are computed using the geodesic distance of samples to the mean or mode. Sampled eigenvectors are  identifiable up to right multiplication by an orthogonal matrix; we resolve this ambiguity by Procrustes aligning all samples to mean or mode prior to computing intervals.

The original Gibbs posterior and the Bayesian spiked covariance model \citep{jauch2020random} are the primary alternatives to the proposed method. The original Gibbs posterior uses $||X-XVV^T||^2$ as a loss function. We compute credible intervals around the mode and tune $\eta$ so the average radius of 95\% credible balls around each component matches the average bootstrapped radius. The Bayesian spiked covariance model assumes the likelihood $x_i\mid V, \Lambda, \sigma^2\sim N\{0, \sigma^2(V\Lambda V^T+I)\}$ with $V\in\mathcal{V}(k,p)$ the eigenvectors, $\Lambda$ a diagonal matrix of positive strictly decreasing eigenvalues, and $\sigma^2>0$ residual noise variance. Priors are chosen as $V\sim 1$, $\lambda_j\sim N(0,5^2)$, $j=1,...,p$, $\sigma^2\sim N(0,5^2)$. Samples are obtained using polar augmentation \citep{jauch2021monte} and Hamiltonian Monte Carlo in Stan \citep{carpenter2017stan}. Credible balls are computed around the mode (estimated with the sample that maximizes the log posterior density) and Frechet mean \citep{chakraborty2019statistics}. Coverage was estimated using 500 data replicates in all cases except the Joint Gibbs and Bayesian spiked covariance models when $p/n=1$, which use only $100$ replicates due to the extreme computational cost of sampling.

Table \ref{table:sims} shows the results. Credible regions around the mode for the Bayesian spiked covariance model have poor coverage in all cases. All other methods perform well when $X$ has Gaussian marginals, with the largest fault being over-coverage of components $4$ and $5$. When $X$ has $t_5$ marginals, the joint Gibbs model significantly under-covers the first two components, and the Bayesian spiked covariance model fails entirely. Both the sequential Gibbs posterior and the bootstrap provide excellent coverage independent of the marginals of $X$ and relative dimension.

\begin{table}
\begin{tabular}{lcccc}
& & $N(0,I)$ &  \\
$p/n$ & $1/4$ & $1/2$ & $1$  \\
Sequential Gibbs & (92, 93, 96, 99, 99) & (91, 94, 97, 98, 99) & (89, 95, 97, 99, 99) \\
Joint Gibbs  & (90, 90, 97, 96, 100) & (88, 93, 96, 96, 100) & (87, 94, 92, 95, 100)  \\
Bootstrap  & (93, 93, 97, 100, 98) & (93, 96, 97, 99, 99) & (91, 96, 98, 99, 100) \\
BPCA (mode)  & (13, 8, 8, 9, 5) & (14, 9, 12, 9, 4) & (16, 13, 14, 9, 3)  \\
BPCA (mean)  & (92, 95, 96, 96, 97) & (90, 95, 96, 98, 99) & (89, 94, 98, 99, 100)  \\[5pt]
& & $t_5(0,I)$ &  \\
$p/n$ & $1/4$ & $1/2$ & $1$ \\
Sequential Gibbs & (94, 89, 91, 94, 97) & (97, 92, 91, 90, 95) & (96, 90, 90, 90, 96) \\
Joint Gibbs  &  (71, 84, 88, 94, 99) & (64, 81, 91, 96, 100) & (63, 82, 79, 97, 100)  \\
Bootstrap  & (95, 89, 92, 94, 96) & (97, 93, 92, 90, 96) & (97, 89, 91, 91, 94)  \\
BPCA (mode)  & (56, 42, 32, 27, 26) & (75, 56, 48, 38, 33) & (86, 67, 62, 45, 40)  \\
BPCA (mean)  & (36, 49, 58, 61, 64) & (19, 31, 39, 46, 57) & (9, 20, 20, 38, 44) \\
\end{tabular}
\vspace*{1em}
\caption{
\textit{Coverage of $95\%$ intervals by component}. Coverage of confidence/credible balls for the first five eigenvectors under different marginal distributions and relative dimensions. BPCA denotes the Bayesian spiked covariance model.
}
\label{table:sims}
\end{table}

\section{Applications to Crime Data}\label{sec:crime}
\subsection{Visualizing Uncertainty}
We analyze the publicly available communities and crime dataset \citep{misc_communities_and_crime_183}, which contains socio-economic, law enforcement, and crime data for communities from the 1990 United States Census, the 1990 United States Law Enforcement Management and Administrative Statistics survey, and the 1995 Federal Bureau of Investigation Uniform Crime Report. We focus on $p=99$ numeric features including median family income, divorce rates, unemployment rates, vacancy rates, number of police officers per capita, and violent crime rate, all normalized to have unit variance. The goal is to identify groups of features predictive of higher violent crime. We applied principal component analysis to the centered/scaled data. Roughly, the first five components capture (1) income and family stability, (2) recent immigration and language barriers, (3) housing availability and occupancy, (4) youth prevalence and neighbourhood age, and (5) homelessness and poverty. These components explain 65\% of the variance in the data; the first 21 components explain 90\% of the variance. Additional information 
is in the supplement.

We subsample $n=100$ communities to illustrate key aspects of uncertainty characterization from \eqref{eq:pca_post}. Figure \ref{fig:scores} shows posterior scores colored by violent crime rate after calibrating with the bootstrap matching algorithm. The variance of the $j$th score vector $[z_{1j},...,z_{nj}]$ increases with $j$. This happens for two reasons. First, uncertainty from previously estimated components accumulates, resulting in higher uncertainty for later components. Second, the eigenvalues of later components are poorly separated compared to the eigenvalues of the first components; this makes it harder to disambiguate directions and results in larger variance, as expected from \Cref{prop:pca_bvm}. The online supplement contains further details on calibration.

\begin{figure}[t]
    \hspace*{-1.6cm}
    \centering
    {\includegraphics[width=18cm]{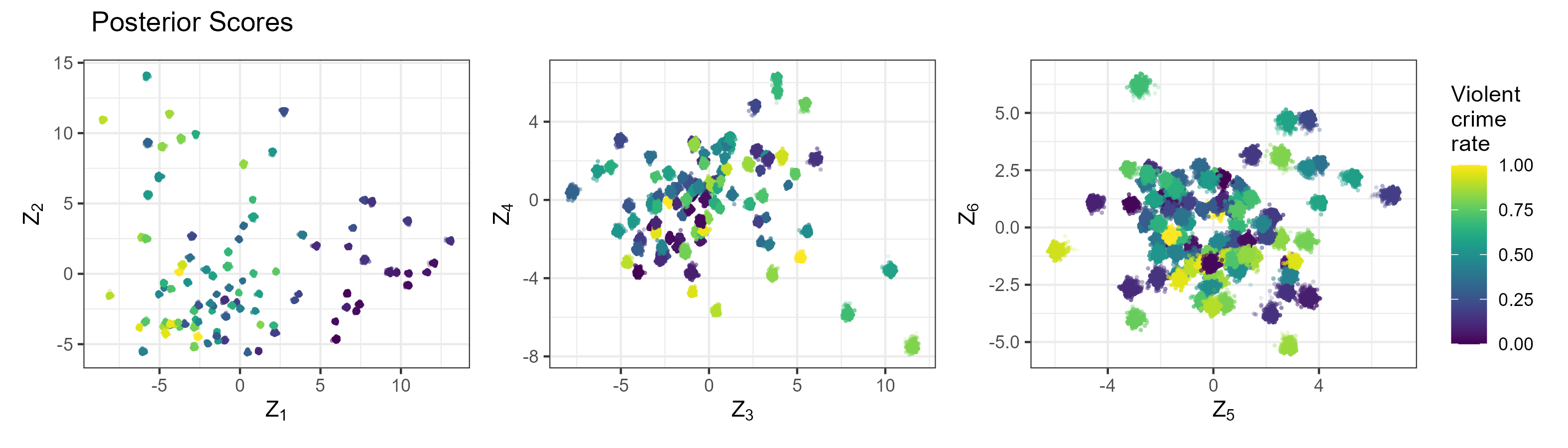}}
    \caption{Posterior scores sampled from the sequential Gibbs posterior, plotted in pairs to illustrate growing uncertainty with component index. Scores are colored by crime rate.}%
    \label{fig:scores}%
\end{figure}
\begin{figure}[t]
    \centering
    {\includegraphics[width=14cm]{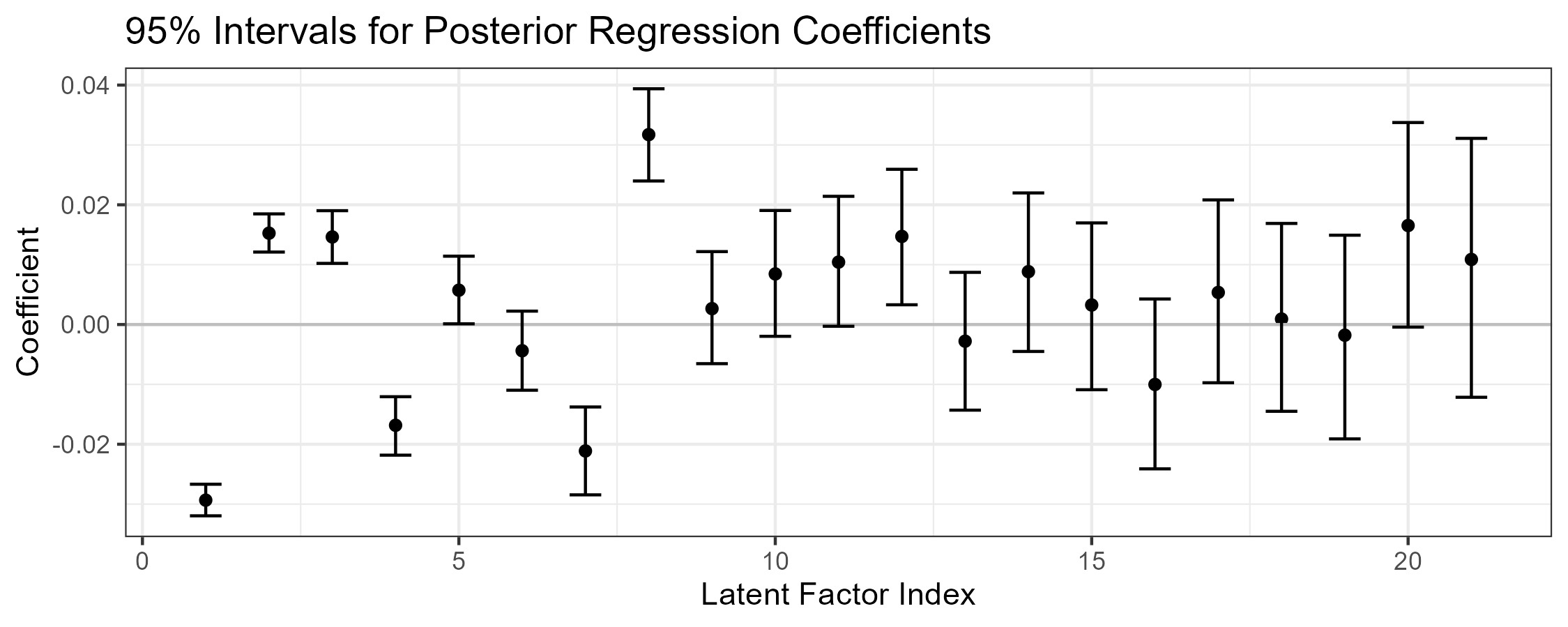}}
    \caption{Posterior credible intervals for coefficients in principal component regression using the sequential posterior. }%
    \label{fig:coeff}%
\end{figure}

\subsection{Principal Component Regression}\label{sec:pcr}
Principal component regression fits a linear model to scores, with $Y=XV\beta+\varepsilon$ where $Y\in\mathbb{R}^n$ is a centered response vector for $n$ individuals, $X\in\mathbb{R}^{n\times p}$ is a centered/scaled matrix of $p$-dimensional features, $V\in\mathcal{V}(J,p)$ are components, $\beta\in\mathbb{R}^J$ are coefficients, and $\varepsilon\in\mathbb{R}^n$ are errors. Adopting the distribution $\varepsilon\sim N(0,\sigma^2 I)$ induces a Gaussian likelihood $\pi(Y\mid V,\beta,\sigma^2)$. We apply our sequential framework to principal component regression, using \eqref{eq:PCA_loss2} for the first $J$ losses and the negative log-likelihood $-\log\{\pi(Y\mid V,\beta,\sigma^2)\}$ for the $J+1$st loss. The scale of the likelihood is well specified relative to priors, so we fix $\eta_{J+1}=1$. When the loss is a negative log-likelihood, \eqref{eq:gibbs} is exactly Bayes' rule. The sequential posterior is
\begin{align}
    \pi_{\eta}(V,\beta,\sigma^2\mid X, Y) &= \kappa_\eta^{(n)}(V\mid X)\pi_{\text{Bayes}}(\beta,\sigma^2\mid X, Y, V) \label{eq:pca_reg}
\end{align}
where $\pi_\text{Bayes}$ is the likelihood-based posterior for $\beta,\sigma^2\mid X,Y,V$ conditional on $V$ and we have parameterized $\kappa_\eta^{(n)}$ from \eqref{eq:pca_post} in terms of $v_j=N_{<j}w_j$. Choosing a normal inverse-gamma prior $\beta\mid\sigma^2 \sim N(0,\sigma^2 I)$, $1/\sigma^2 \sim Ga(1,1)$ results in a conjugate posterior for $\pi_{\text{Bayes}}$ and allows exact sampling of \eqref{eq:pca_reg}. 

We apply \eqref{eq:pca_reg} to the communities and crime dataset. Figure \ref{fig:coeff} shows posterior credible intervals for coefficients. The first eight components are significant, and the results are largely intuitive: for example, violent crime decreases as community income and family stability increases. As before, uncertainty grows with the score index, resulting in wider credible intervals for later coefficients. Additional analysis may be found in the online supplement.

\section{Discussion}
Sequential Gibbs posteriors introduce many potential applications and research directions. One area of interest is combining loss-based Gibbs posteriors with traditional likelihood-based posteriors, as illustrated in \Cref{sec:pcr}. This arises when some parameters are characterized by a likelihood and others by a non-likelihood-based loss. For example, we may use a machine learning algorithm, such as a neural network, for dimensionality reduction for complex high-dimensional features, but then use a likelihood for a low-dimensional response. In addition to improving robustness, this may have major computational advantages over attempting likelihood-based neural network inferences.

Sequential Gibbs posteriors apply to a wide range of loss functions and problems not discussed in this work. It is interesting to extend our principal component analysis results to variants such as sparse, functional, and disjoint principal component analysis. Beyond principal component analysis, sequential Gibbs posteriors can be applied to specific problems in the general settings detailed in Examples \ref{ex:multiscale} and \ref{ex:matrix} as well as to neural networks as just described. \textcolor{black}{Section S3.5 of the supplement contains a blueprint for applying sequential posteriors in generic hierarchical models, including details for the special case of generalized linear regression with random effects.} Nonlinear dimension reduction methods such as diffusion maps may also benefit from sequential Gibbs posteriors since they, like principal component analysis, rely on eigenvectors of matrices built from data and are often used to process data prior to further analyses such as regression. In particular, sequential Gibbs posteriors can provide uncertainty quantification in these settings.

\textcolor{black}{Extensions to general manifolds are also interesting -- for example, those defined by constraint functions $\mathcal{M} = \{\theta\mid g(\theta) \geq 0\}$ for some function $g$. This may be achieved by extending sequential Gibbs posteriors to allow for a joint prior over all parameters that does not necessarily factor. We expect this object retains the key asymptotic properties developed in our work, although the current induction-based proof techniques do not apply.}

\textcolor{black}{Another line of future work is calibration of the hyperparameters $\eta=(\eta_1,\dots,\eta_J)$. In particular, it is desirable to have theoretical results guaranteeing appropriate coverage. For non-Euclidean parameters this may require development of bootstrap theory for confidence balls on general manifolds. It also remains to be seen how calibration of $\eta$ relates to selection of penalty parameters, for example in the context of sparse principal component analysis and when a regularization penalty is applied to the neural network loss mentioned above.}

\section*{Acknowledgements}
This work was partially funded by grants from the United States Office of Naval Research (N000142112510) and National Institutes of Health (R01ES028804, R01ES035625).


\bibliographystyle{abbrv}
\bibliography{paper-ref}

\begin{thebibliography}{10}

\bibitem{agnoletto2025bayesian}
D.~Agnoletto, T.~Rigon, and D.~B. Dunson.
\newblock Bayesian inference for generalized linear models via quasi-posteriors.
\newblock {\em Biometrika}, page asaf022, 2025.

\bibitem{betancourt2013general}
M.~Betancourt.
\newblock A general metric for riemannian manifold hamiltonian monte carlo.
\newblock In {\em International Conference on Geometric Science of Information}, pages 327--334. Springer, 2013.

\bibitem{bhattacharya2022gibbs}
I.~Bhattacharya and R.~Martin.
\newblock Gibbs posterior inference on multivariate quantiles.
\newblock {\em Journal of Statistical Planning and Inference}, 218:106--121, 2022.

\bibitem{bhattacharya2014statistics}
R.~Bhattacharya and V.~Patrangenaru.
\newblock Statistics on manifolds and landmarks based image analysis: A nonparametric theory with applications.
\newblock {\em Journal of Statistical Planning and Inference}, 145:1--22, 2014.

\bibitem{bissiri2016general}
P.~G. Bissiri, C.~C. Holmes, and S.~G. Walker.
\newblock A general framework for updating belief distributions.
\newblock {\em Journal of the Royal Statistical Society Series B: Statistical Methodology}, 78(5):1103--1130, 2016.

\bibitem{braides2006handbook}
A.~Braides.
\newblock A handbook of $\gamma$-convergence.
\newblock In {\em Handbook of Differential Equations: Stationary Partial Differential Equations}, volume~3, pages 101--213. Elsevier, 2006.

\bibitem{carpenter2017stan}
B.~Carpenter, A.~Gelman, M.~D. Hoffman, D.~Lee, B.~Goodrich, M.~Betancourt, M.~A. Brubaker, J.~Guo, P.~Li, and A.~Riddell.
\newblock Stan: A probabilistic programming language.
\newblock {\em Journal of Statistical Software}, 76, 2017.

\bibitem{chakraborty2019statistics}
R.~Chakraborty and B.~C. Vemuri.
\newblock Statistics on the {S}tiefel manifold: {T}heory and applications.
\newblock {\em Annals of Statistics}, 47(1):415--438, 2019.

\bibitem{chen1999strong}
K.~Chen, I.~Hu, and Z.~Ying.
\newblock Strong consistency of maximum quasi-likelihood estimators in generalized linear models with fixed and adaptive designs.
\newblock {\em The Annals of Statistics}, 27(4):1155--1163, 1999.

\bibitem{ciobotaru2022consistency}
C.~Ciobotaru and C.~Mazza.
\newblock Consistency and asymptotic normality of {M}-estimates of scatter on {G}rassmann manifolds.
\newblock {\em Journal of Multivariate Analysis}, 190:104998, 2022.

\bibitem{dziugaite2017computing}
G.~K. Dziugaite and D.~M. Roy.
\newblock Computing nonvacuous generalization bounds for deep (stochastic) neural networks with many more parameters than training data.
\newblock {\em Proceedings of the Thirty-Third Conference on Uncertainty in Artificial Intelligence}, 33, 2017.

\bibitem{eltzner2019smeary}
B.~Eltzner and S.~F. Huckemann.
\newblock A smeary central limit theorem for manifolds with application to high-dimensional spheres.
\newblock {\em The Annals of Statistics}, 47(6):3360--3381, 2019.

\bibitem{elvira2017bayesian}
C.~Elvira, P.~Chainais, and N.~Dobigeon.
\newblock Bayesian nonparametric subspace estimation.
\newblock In {\em 2017 IEEE International Conference on Acoustics, Speech and Signal Processing (ICASSP)}, pages 2247--2251. IEEE, 2017.

\bibitem{fox2012multiresolution}
E.~Fox and D.~Dunson.
\newblock Multiresolution {G}aussian processes.
\newblock {\em Advances in Neural Information Processing Systems}, 25, 2012.

\bibitem{germain2009pac}
P.~Germain, A.~Lacasse, F.~Laviolette, and M.~Marchand.
\newblock P{A}{C}-{B}ayesian learning of linear classifiers.
\newblock In {\em Proceedings of the 26th Annual International Conference on Machine Learning}, volume~26, pages 353--360, 2009.

\bibitem{girolami2011riemann}
M.~Girolami and B.~Calderhead.
\newblock Riemann manifold langevin and hamiltonian monte carlo methods.
\newblock {\em Journal of the Royal Statistical Society Series B: Statistical Methodology}, 73(2):123--214, 2011.

\bibitem{grunwald2017inconsistency}
P.~Gr{\"u}nwald and T.~Van~Ommen.
\newblock Inconsistency of {B}ayesian inference for misspecified linear models, and a proposal for repairing it.
\newblock {\em Bayesian Analysis}, 12(4):1069--1103, 2017.

\bibitem{guedj2019primer}
B.~Guedj.
\newblock A primer on {P}{A}{C}-{B}ayesian learning.
\newblock In {\em Proceedings of the French Mathematical Society}, volume~33, pages 391--414. Soci{\'e}t{\'e} Math{\'e}matique de France, 2019.

\bibitem{hall1990asymptotic}
P.~Hall.
\newblock Asymptotic properties of the bootstrap for heavy-tailed distributions.
\newblock {\em The Annals of Probability}, 18(3):1342--1360, 1990.

\bibitem{hall2013bootstrap}
P.~Hall.
\newblock {\em The Bootstrap and Edgeworth Expansion}, volume~1.
\newblock Springer Science \& Business Media, 2013.

\bibitem{hernandez2017general}
D.~Hernandez-Stumpfhauser, F.~J. Breidt, and M.~J. van~der Woerd.
\newblock The general projected normal distribution of arbitrary dimension: {M}odeling and {B}ayesian inference.
\newblock {\em Bayesian Analysis}, 12(1), 2017.

\bibitem{hoff2009simulation}
P.~D. Hoff.
\newblock Simulation of the matrix {B}ingham--von {M}ises--{F}isher distribution, with applications to multivariate and relational data.
\newblock {\em Journal of Computational and Graphical Statistics}, 18(2):438--456, 2009.

\bibitem{holbrook2016bayesian}
A.~Holbrook, A.~Vandenberg-Rodes, and B.~Shahbaba.
\newblock Bayesian inference on matrix manifolds for linear dimensionality reduction.
\newblock {\em arXiv preprint arXiv:1606.04478}, 2016.

\bibitem{holmes2017assigning}
C.~C. Holmes and S.~G. Walker.
\newblock Assigning a value to a power likelihood in a general {B}ayesian model.
\newblock {\em Biometrika}, 104(2):497--503, 2017.

\bibitem{hout2013multidimensional}
M.~C. Hout, M.~H. Papesh, and S.~D. Goldinger.
\newblock Multidimensional scaling.
\newblock {\em Wiley Interdisciplinary Reviews: Cognitive Science}, 4(1):93--103, 2013.

\bibitem{jacob2020unbiased}
P.~E. Jacob, J.~O’leary, and Y.~F. Atchad{\'e}.
\newblock Unbiased markov chain monte carlo methods with couplings.
\newblock {\em Journal of the Royal Statistical Society Series B: Statistical Methodology}, 82(3):543--600, 2020.

\bibitem{jauch2020random}
M.~Jauch, P.~D. Hoff, and D.~B. Dunson.
\newblock Random orthogonal matrices and the {C}ayley transform.
\newblock {\em Bernoulli}, 26(2):1560--1586, 2020.

\bibitem{jauch2021monte}
M.~Jauch, P.~D. Hoff, and D.~B. Dunson.
\newblock Monte {C}arlo simulation on the {S}tiefel manifold via polar expansion.
\newblock {\em Journal of Computational and Graphical Statistics}, 30(3):622--631, 2021.

\bibitem{jiang2008gibbs}
W.~Jiang and M.~A. Tanner.
\newblock Gibbs posterior for variable selection in high-dimensional classification and data mining.
\newblock {\em The Annals of Statistics}, 36(5):2207--2231, 2008.

\bibitem{karoui2016can}
N.~E. Karoui and E.~Purdom.
\newblock Can we trust the bootstrap in high-dimension?
\newblock {\em Journal of Machine Learning Research}, 19:1--66, 2016.

\bibitem{katzfuss2017multi}
M.~Katzfuss.
\newblock A multi-resolution approximation for massive spatial datasets.
\newblock {\em Journal of the American Statistical Association}, 112(517):201--214, 2017.

\bibitem{kendall2011limit}
W.~S. Kendall and H.~Le.
\newblock Limit theorems for empirical {F}r{\'e}chet means of independent and non-identically distributed manifold-valued random variables.
\newblock {\em Brazilian Journal of Probability and Statistics}, 25(3):323--352, 2011.

\bibitem{kent2013new}
J.~T. Kent, A.~M. Ganeiber, and K.~V. Mardia.
\newblock A new method to simulate the {B}ingham and related distributions in directional data analysis with applications.
\newblock {\em arXiv preprint arXiv:1310.8110}, 2013.

\bibitem{kolda2009tensor}
T.~G. Kolda and B.~W. Bader.
\newblock Tensor decompositions and applications.
\newblock {\em SIAM Review}, 51(3):455--500, 2009.

\bibitem{kysely2008cautionary}
J.~Kysel{\`y}.
\newblock A cautionary note on the use of nonparametric bootstrap for estimating uncertainties in extreme-value models.
\newblock {\em Journal of Applied Meteorology and Climatology}, 47(12):3236--3251, 2008.

\bibitem{lee1999learning}
D.~D. Lee and H.~S. Seung.
\newblock Learning the parts of objects by non-negative matrix factorization.
\newblock {\em Nature}, 401(6755):788--791, 1999.

\bibitem{lee2012smooth}
J.~M. Lee.
\newblock {\em Smooth manifolds}, volume~1.
\newblock Springer, 2012.

\bibitem{lin2017bayesian}
L.~Lin, V.~Rao, and D.~Dunson.
\newblock Bayesian nonparametric inference on the {S}tiefel manifold.
\newblock {\em Statistica Sinica}, 27(2):535--553, 2017.

\bibitem{lyddon2019general}
S.~P. Lyddon, C.~Holmes, and S.~Walker.
\newblock General {B}ayesian updating and the loss-likelihood bootstrap.
\newblock {\em Biometrika}, 106(2):465--478, 2019.

\bibitem{magnus1985differentiating}
J.~R. Magnus.
\newblock On differentiating eigenvalues and eigenvectors.
\newblock {\em Econometric Theory}, 1(2):179--191, 1985.

\bibitem{martin2022direct}
R.~Martin and N.~Syring.
\newblock Direct {G}ibbs posterior inference on risk minimizers: construction, concentration, and calibration.
\newblock {\em Handbook of Statistics}, 47:1--47, 2022.

\bibitem{miller2021asymptotic}
J.~W. Miller.
\newblock Asymptotic normality, concentration, and coverage of generalized posteriors.
\newblock {\em The Journal of Machine Learning Research}, 22(1):7598--7650, 2021.

\bibitem{nakatsukasa2022fast}
Y.~Nakatsukasa and T.~Park.
\newblock A fast randomized algorithm for computing the null space.
\newblock {\em BIT Numerical Mathematics}, 63(36), 2022.

\bibitem{neyshabur2017exploring}
B.~Neyshabur, S.~Bhojanapalli, D.~McAllester, and N.~Srebro.
\newblock Exploring generalization in deep learning.
\newblock {\em Advances in Neural Information Processing Systems}, 30, 2017.

\bibitem{ng2001spectral}
A.~Ng, M.~Jordan, and Y.~Weiss.
\newblock On spectral clustering: {A}nalysis and an algorithm.
\newblock {\em Advances in Neural Information Processing Systems}, 14, 2001.

\bibitem{nychka2015multiresolution}
D.~Nychka, S.~Bandyopadhyay, D.~Hammerling, F.~Lindgren, and S.~Sain.
\newblock A multiresolution {G}aussian process model for the analysis of large spatial datasets.
\newblock {\em Journal of Computational and Graphical Statistics}, 24(2):579--599, 2015.

\bibitem{ogasawara2002concise}
H.~Ogasawara.
\newblock Concise formulas for the standard errors of component loading estimates.
\newblock {\em Psychometrika}, 67:289--297, 2002.

\bibitem{paindaveine2020inference}
D.~Paindaveine and T.~Verdebout.
\newblock Inference for spherical location under high concentration.
\newblock {\em The Annals of Statistics}, 48(5):2982--2998, 2020.

\bibitem{peruzzi2018bayesian}
M.~Peruzzi and D.~B. Dunson.
\newblock Bayesian modular and multiscale regression.
\newblock {\em arXiv preprint arXiv:1809.05935}, 2018.

\bibitem{plummer2015cuts}
M.~Plummer.
\newblock Cuts in bayesian graphical models.
\newblock {\em Statistics and Computing}, 25:37--43, 2015.

\bibitem{potscher2009distribution}
B.~M. P{\"o}tscher and H.~Leeb.
\newblock On the distribution of penalized maximum likelihood estimators: {T}he {LASSO}, {SCAD}, and thresholding.
\newblock {\em Journal of Multivariate Analysis}, 100(9):2065--2082, 2009.

\bibitem{pourzanjani2021bayesian}
A.~A. Pourzanjani, R.~M. Jiang, B.~Mitchell, P.~J. Atzberger, and L.~R. Petzold.
\newblock Bayesian inference over the {S}tiefel manifold via the {G}ivens representation.
\newblock {\em Bayesian Analysis}, 16(2):639--666, 2021.

\bibitem{misc_communities_and_crime_183}
M.~Redmond.
\newblock {Communities and Crime}.
\newblock UCI Machine Learning Repository, 2009.
\newblock {DOI}: https://doi.org/10.24432/C53W3X.

\bibitem{rigon2023generalized}
T.~Rigon, A.~H. Herring, and D.~B. Dunson.
\newblock A generalized {B}ayes framework for probabilistic clustering.
\newblock {\em Biometrika}, 110(3):559--578, 2023.

\bibitem{robbins1951stochastic}
H.~Robbins and S.~Monro.
\newblock A stochastic approximation method.
\newblock {\em The Annals of Mathematical Statistics}, 22(1):400--407, 1951.

\bibitem{schenker1985qualms}
N.~Schenker.
\newblock Qualms about bootstrap confidence intervals.
\newblock {\em Journal of the American Statistical Association}, 80(390):360--361, 1985.

\bibitem{schmidler2007fast}
S.~C. Schmidler.
\newblock Fast {B}ayesian shape matching using geometric algorithms.
\newblock {\em Bayesian statistics}, 8:471--490, 2007.

\bibitem{sethuraman1961some}
J.~Sethuraman.
\newblock Some limit theorems for joint distributions.
\newblock {\em Sankhy{\=a}: The Indian Journal of Statistics, Series A}, 23(4):379--386, 1961.

\bibitem{syring2019calibrating}
N.~Syring and R.~Martin.
\newblock Calibrating general posterior credible regions.
\newblock {\em Biometrika}, 106(2):479--486, 2019.

\bibitem{syring2020robust}
N.~Syring and R.~Martin.
\newblock Robust and rate-optimal {G}ibbs posterior inference on the boundary of a noisy image.
\newblock {\em Annals of Statistics}, 48(3), 2020.

\bibitem{thiemann2017strongly}
N.~Thiemann, C.~Igel, O.~Wintenberger, and Y.~Seldin.
\newblock A strongly quasiconvex {P}{A}{C}-{B}ayesian bound.
\newblock In {\em International Conference on Algorithmic Learning Theory}, volume~20, pages 466--492, 2017.

\bibitem{thomas2022learning}
B.~S. Thomas, K.~You, L.~Lin, L.-H. Lim, and S.~Mukherjee.
\newblock Learning subspaces of different dimensions.
\newblock {\em Journal of Computational and Graphical Statistics}, 31(2):337--350, 2022.

\bibitem{van2014asymptotically}
S.~Van~de Geer, P.~B{\"u}hlmann, Y.~Ritov, and R.~Dezeure.
\newblock On asymptotically optimal confidence regions and tests for high-dimensional models.
\newblock {\em Project Euclid}, 42(3):1166--1202, 2014.

\bibitem{wang2013adaptive}
Z.~Wang, S.~Mohamed, and N.~Freitas.
\newblock Adaptive hamiltonian and riemann manifold monte carlo.
\newblock In {\em International Conference on Machine Learning}, pages 1462--1470. PMLR, 2013.

\bibitem{wedderburn1974quasi}
R.~W. Wedderburn.
\newblock Quasi-likelihood functions, generalized linear models, and the gauss—newton method.
\newblock {\em Biometrika}, 61(3):439--447, 1974.

\bibitem{wold1966estimation}
H.~Wold.
\newblock Estimation of principal components and related models by iterative least squares.
\newblock {\em Multivariate Analysis}, 36(5):391--420, 1966.

\bibitem{wu2023comparison}
P.-S. Wu and R.~Martin.
\newblock A comparison of learning rate selection methods in generalized {B}ayesian inference.
\newblock {\em Bayesian Analysis}, 18(1):105--132, 2023.

\bibitem{zhang2014confidence}
C.-H. Zhang and S.~S. Zhang.
\newblock Confidence intervals for low dimensional parameters in high dimensional linear models.
\newblock {\em Journal of the Royal Statistical Society Series B: Statistical Methodology}, 76(1):217--242, 2014.

\bibitem{zhang2006a}
T.~Zhang.
\newblock From $\varepsilon$-entropy to {K}{L}-entropy: Analysis of minimum information complexity density estimation.
\newblock {\em The Annals of Statistics}, 34(5):2180--2210, 2006.

\bibitem{zhang2006b}
T.~Zhang.
\newblock Information-theoretic upper and lower bounds for statistical estimation.
\newblock {\em IEEE Transactions on Information Theory}, 52(4):1307--1321, 2006.

\end{thebibliography}


\newpage


\appendix

\section*{Summary}\label{sec:summary}

This section contains supplementary materials for ``\textit{Sequential Gibbs posteriors with applications to principal component analysis}." In \Cref{sec:proofs} we list additional assumptions and prove Theorems \ref{thrm:bvm}, \ref{thrm:concentration}, \ref{thrm:sequential}, and \Cref{prop:pca_bvm} from the main text. In \Cref{sec:geometry} we discuss derivatives on manifolds and prove all manifold-related assumptions in the text and \Cref{sec:proofs} are well-defined in the sense of being chart-invariant. In \Cref{sec:simulations} we detail our simulations, and in \Cref{sec:crime_sup} we expand on our application to the crime dataset in \Cref{sec:crime} of the main text.

\section{Proofs}\label{sec:proofs}

Proofs and additional assumptions for Theorems \ref{thrm:bvm}, \ref{thrm:concentration}, and \ref{thrm:sequential} and the proof of \Cref{prop:pca_bvm} are in Sections \ref{sec:bvm_proof}, \ref{sec:concentration_proof}, \ref{sec:sequential_bvm_proof}, and \ref{sec:prop1}, respectively. Without loss of generality all proofs assume $\eta_1=\cdots=\eta_J=1$. All assumptions are to be interpreted as holding almost surely and any assumptions stated for a single $j$ implicitly hold for all $j\in[J]$. Since its presence is implied by the sample size $n$, the data variable $x$ is henceforth omitted from notation. We also define
\begin{align}
	\Pi_j^{(n)}(d\theta_j\mid \theta_{<j}) &= \frac{1}{z_j^{(n)}(\theta_{<j})}\exp\{-\eta_jn\ell_j^{(n)}(\theta_j\mid\theta_{<j})\} \Pi^{(0)}_j(d\theta_j)\label{eq:cond}
\end{align}
so that $\Pi^{(n)}(d\theta)=\prod_{j=1}^J\Pi^{(n)}_j(d\theta_j\mid\theta_{<j})$.

\subsection{Proof of \Cref{thrm:bvm}}\label{sec:bvm_proof}

\Cref{thrm:bvm} extends asymptotic normality results from \cite{miller2021asymptotic} to manifolds. The following are the additional assumptions for \Cref{thrm:bvm}, which are manifold analogues of the assumptions in \cite[Theorem 5]{miller2021asymptotic}. Recall $\mathcal{M}$ is a smooth $p$-dimensional manifold and a chart on $\mathcal{M}$ is a pair $(U,\varphi)$ where $U\subseteq\mathcal{M}$ is open and $\varphi:U\to\varphi(U)$ is a smooth diffeomorphism. 

\begin{assumption}[Uniformly bounded third derivatives]\label{as:third}
There is an open, bounded $E\subseteq\mathcal{M}$ and chart $(V,\psi)$ with $\phi^\star\in E\cap V$ such that $\ell^{(n)}$ has continuous third derivatives on $E$ and
\begin{align}\label{eq:bound}
    \sup_n\sup_{\theta\in E\cap V}\sup_{i,j,k}\lvert\partial_{ijk} (\ell^{(n)}\circ\psi^{-1})\{\psi(\theta)\}\rvert &< \infty.
\end{align}
\end{assumption}

\begin{assumption}[Positive definite Hessian]\label{as:hessian}
The Hessian $\ell''(\phi^\star)$ is positive definite.
\end{assumption}

In \Cref{as:hessian} ``$\ell''(\phi^\star)$ is positive definite" means there exists a chart $(U,\varphi)$ on $\mathcal{M}$ containing $\phi^\star$ such that the Hessian of $\ell\circ\varphi^{-1}$ is positive-definite at $\varphi(\phi^\star)$. Second and third derivatives on manifolds are not chart-invariant in general. However, \Cref{lem:invariant} in \Cref{sec:geometry} says the above conditions are chart-invariant, and hence well-defined, justifying the use of local coordinates throughout this work. It also says \Cref{as:third} implies $\ell\circ\varphi^{-1}$ is twice differentiable in a neighborhood of $\phi^\star$ for any chart $(U,\varphi)$ containing $\phi^\star$. Thus \Cref{as:hessian} is well-defined.

\begin{proof}[of \Cref{thrm:bvm}]
The proof proceeds by mapping all quantities to Euclidean space and applying \cite[Theorem 5]{miller2021asymptotic}. Euclidean objects are distinguished with a tilde. Fix a chart $(U,\varphi)$ \textcolor{black}{satisfying $\mathrm{supp}(\Pi^{(0)}) \subseteq U$ and}, shrinking $K$ and $E$ if necessary, assume $K\subseteq E\subseteq U$ with $K$ compact and $\phi^\star$ in its interior, and $E$ satisfying \Cref{as:third}. Set $\tilde\phi^\star=\varphi(\phi^\star)$, $\tilde \ell^{(n)}=\ell^{(n)}\circ\varphi^{-1}$, $\tilde \ell=\ell\circ\varphi^{-1}$, $\tilde U=\varphi(U)$, $\tilde E=\varphi(E)$, and $\tilde K=\varphi(K)$. Then $\tilde K\subseteq\tilde E\subseteq\tilde U$ and, since $\varphi$ is a diffeomorphism, $\tilde E$ is open and bounded in $\mathbb{R}^p$, $\tilde K$ is compact, $\tilde\phi^\star$ is in its interior, $\tilde\ell^{(n)}\to\tilde\ell$ almost surely, and $\tilde\ell^{(n)}$ have continuous third derivatives on $\tilde E$. Furthermore, by \Cref{lem:invariant}, the collection $\{\tilde\ell^{(n)'''} : n\in\mathbb{N}\}$ is uniformly bounded on $E$. If $\tilde\theta\in\tilde K\setminus\{\tilde\phi^\star\}$ then $\tilde\theta=\varphi(\theta)$ for some $\theta\in K\setminus\{\phi^\star\}$. Thus, since $\ell(\theta) > \ell(\phi^\star)$ by Assumption (a) in the statement of the theorem, we have $\tilde \ell(\tilde\theta) = \ell[\varphi^{-1}\{\varphi(\theta)\}] = \ell(\theta) > \ell(\phi^\star) = \tilde \ell(\tilde\phi^\star)$. Similarly, if $\tilde\theta\in \tilde U\setminus\tilde K$ then $\tilde\theta=\varphi(\theta)$ for some $\theta\in U\setminus K$. So $\tilde \ell^{(n)}(\tilde\theta)=\ell^{(n)}(\theta)$ and, again by Assumption (a), 
\begin{align*}
    \liminf_n \inf_{\tilde\theta \in \tilde U\setminus \tilde K} \{\tilde \ell^{(n)}(\tilde \theta)-\tilde \ell(\tilde\phi^\star)\} > 0.
\end{align*}
By \Cref{lem:invariant}, $\ell''$ is well-defined and chart-invariant at $\phi^\star$, and $\tilde \ell''(\tilde\phi^\star)=\ell''(\phi^\star)$ is positive definite by \Cref{as:hessian}. \textcolor{black}{This concludes verification that the Euclidean objects $\tilde\phi^\star$, $\tilde \ell^{(n)}$, $\tilde \ell$, $\tilde U$, $\tilde E$, and $\tilde K$ satisfy the assumptions in \cite[Theorem 5]{miller2021asymptotic}.} 

\textcolor{black}{Let $\pi^{(n)}$ be the density of the Gibbs posterior, $\Pi^{(n)}$. The pushforward $\tilde\pi^{(n)}=\varphi_\#\pi^{(n)}$ is well-defined because $\mathrm{supp}(\Pi^{(0)}) \subseteq U$ implies $\mathrm{supp}(\Pi^{(n)}) \subseteq U$, and hence that $\pi^{(n)}$ is a valid probability density on $U$ for all $n$. Furthermore, by change of variables,
\begin{align*}
	\tilde\pi^{(n)}(\tilde\theta) &=
		\frac{1}{z^{(n)}}\exp\{-n\tilde \ell^{(n)}(\tilde\theta)\}\tilde\pi^{(0)}(\tilde\theta)
\end{align*}
where $\tilde\pi^{(0)}(\tilde\theta)=(\pi^{(0)}\circ\varphi^{-1})(\tilde\theta)\lvert\det (\varphi^{-1})'(\tilde\theta)\rvert=\varphi_\#\pi^{(0)}(\tilde\theta)$. Since $\varphi$ is a diffeomophism, $\tilde\pi^{(0)}$ is continuous and strictly positive at $\tilde\phi^\star$. Thus, in addition to the Euclidean objects in the preceding paragraph, the prior $\tilde\pi^{(0)}$ also satisfies the assumptions in \cite[Theorem 5]{miller2021asymptotic}. Therefore, by that result and the fact that $(f\circ g)_\# = f_\#g_\#$, we have 
\begin{align*}
    d_{TV}\{(\tau^{(n)}\circ\varphi)_\#\Pi^{(n)}, N(0, H_\varphi^{-1})\} &= d_{TV}\{\tau_\#^{(n)}\tilde\Pi^{(n)}, N(0, H_\varphi^{-1})\}
        \to 0
\end{align*}
where $H_\varphi=\tilde \ell''(\tilde\phi^\star)$ and $\tau^{(n)}(\tilde\theta)=\sqrt{n}(\tilde\theta-\tilde\theta^{(n)})$.}
\end{proof}

\textcolor{black}{With \Cref{thrm:bvm} established, we now prove \eqref{eq:bvm_alternative} holds when the assumption $\mathrm{supp}(\Pi^{(0)})\subseteq U$ is replaced with $\varphi(U)=\mathbb{R}^p$. The proof uses the following lemma.}

\begin{lemma}[Total variation and truncation]\label{lem:tv}
    Let $X_n$ be a random variable with $P(X_n\in U^c)\to 0$ for some set $U$ and define $Y_n = X_n\mid X_n\in U$. Then $d_{TV}(X_n, Y_n)\to 0$.
\end{lemma}

\begin{proof}
    We calculate
    \begin{align*}
        d_{TV}(X_n, Y_n) &= \sup_A |P(X_n\in A)- P(Y_n \in A)|  \\
        &=  \sup_A \bigg|P(X_n\in A)- \frac{P(X_n\in A\cap U)}{P(X_n\in U)}\bigg| \\
        &= \sup_A \bigg|P(X_n\in A)- P(X_n\in A\cap U) + P(X_n\in A\cap U) - \frac{P(X_n\in A\cap U)}{P(X_n\in U)}\bigg| \\
        &\leq \sup_A |P(X_n\in A)- P(X_n\in A\cap U)| \\
        &\qquad + \sup_A \bigg|P(X_n\in A\cap U) - \frac{P(X_n\in A\cap U)}{P(X_n\in U)}\bigg| \\
        &= \sup_A |P(X_n\in A\cap U^c)| +\bigg\{1-\frac{1}{P(X_n\in U)}\bigg\}\sup_A |P(X_n\in A\cap U)| \\
        &\leq P(X_n\in U^c)+1-\frac{1}{P(X_n\in U)}
    \end{align*}
    which vanishes.
\end{proof}

\noindent\textcolor{black}{\textit{Proof of \Cref{eq:bvm_alternative}}. The assumption $\mathrm{supp}(\Pi^{(0)})\subseteq U$ implies $\tilde\pi^{(n)}$ is supported on $U$ and -- since $\varphi$ is a diffeomorphism on $U$ rather than all of $\mathcal{M}$ -- justifies the use of change of variables in the proof of \Cref{thrm:bvm}. Suppose now that $(U,\varphi)$ is a chart on $\mathcal{M}$ such that $\varphi(U)=\mathbb{R}^p$, but $\mathrm{supp}(\Pi^{(0)})\not\subseteq U$. Let $\Pi^{(n)}_U$ be the restricted Gibbs posterior with density
\begin{align*}
	\pi_{U}^{(n)}(\theta) &= \frac{1}{z_{U}^{(n)}}\exp\{-n\ell^{(n)}(\theta)\}\pi^{(0)}_U(\theta),
\end{align*}
where $\pi^{(0)}_U(\theta)=\pi_{0}(\theta)1_U(\theta)$ and $z_{U}^{(n)} = \int_U \exp\{-n\ell^{(n)}(\theta)\}\pi^{(0)}_U(\theta)d\theta$. The normalizing constant $z_U^{(n)}$ is nonzero because we still assume $\phi^\star\in U$ and $\pi^{(0)}$ is continuous and strictly positive at $\phi^\star$. Since the prior density $\pi^{(0)}_U$ is supported in $U$, we have by \Cref{thrm:bvm} that
\begin{align*}
    d_{TV}\{\Pi_U^{(n)}, (\tau^{(n)}\circ\varphi)^{-1}_\#N(0, H_\varphi^{-1})\} &= d_{TV}\{(\tau^{(n)}\circ\varphi)_\#\Pi_U^{(n)}, N(0, H^{-1})\}
        \to 0,
\end{align*}
the equality holding since both $\varphi$ and $\tau^{(n)}$ are invertible. Next, by \Cref{thrm:concentration} and \Cref{lem:tv}, $d_{TV}(\Pi_U^{(n)},\Pi^{(n)})\to 0$. So by the triangle inequality,
\begin{align*}
\scalebox{0.95}{%
    $d_{TV}\{\Pi^{(n)}, (\tau^{(n)}\circ\varphi)^{-1}_\#N(0,H_\varphi^{-1})\} \leq 
    d_{TV}(\Pi^{(n)},\Pi_U^{(n)}) + 
    d_{TV}\{\Pi_U^{(n)}, (\tau^{(n)}\circ\varphi)^{-1}_\#N(0, H_\varphi^{-1})\},$
}
\end{align*}
which vanishes as $n\to\infty$. The condition $\varphi(U)=\mathbb{R}^p$ ensures that $(\tau^{(n)}\circ\varphi)^{-1}_\#N(0,H_\varphi^{-1})$ is a valid probability measure since $N(0,H_\varphi^{-1})$ is a probability measure on all of $\mathbb{R}^p$. Thus, under the same assumptions as in \Cref{thrm:bvm} but with $\varphi(U)=\mathbb{R}^p$ instead of $\mathrm{supp}(\Pi^{(0)})\subseteq U$, we have proven that $d_{TV}\{\Pi^{(n)}, (\tau^{(n)}\circ\varphi)^{-1}_\#N(0,H_\varphi^{-1})\}\to 0$.}

\subsection{Proof of \Cref{thrm:concentration}}\label{sec:concentration_proof}

Recall that $d_j$ is a metric on $\mathcal{M}_j$, $d$ is the metric on $\mathcal{M}$ given by $d^2=d_1^2+\cdots+d_J^2$, $N_{j,\epsilon}=\{\theta_j:d_j(\theta_j,\phi^\star_j)<\epsilon\}$, and $N_\epsilon=\{\theta:d(\theta,\phi^\star)<\epsilon\}$. We also let $d_{<j}$ denote the metric on $\mathcal{M}_{<J}$ given by $d_{<j}^2 = d_1^2+\cdots+d_{j-1}^2$ and set $N_{<j,\epsilon}=\{\theta_{<j}:d_{<j}(\theta_{<j},\phi^\star_{<j})<\epsilon\}$. \textcolor{black}{In general, for smooth manifolds $\mathcal{M}$ and $\mathcal{N}$, which are topological spaces by definition, a function $f:\mathcal{M}\to\mathcal{N}$ is continuous if $f^{-1}(U)$ is open in $\mathcal{M}$ for every open $U\subseteq\mathcal{N}$. As this is the case in what follows, when $\mathcal{M}$ and $\mathcal{N}$ are equipped with metrics $d_\mathcal{M}$ and $d_\mathcal{N}$, continuity of $f$ is equivalent to the statement: If $d_\mathcal{M}(x_n,x)\to 0$, then $d_\mathcal{N}\{f(x_n),f(x)\}\to 0$.}


\begin{assumption}\label{as:concentration}
In the notation of \Cref{thrm:concentration}, for each $j\in[J]$:
\begin{enumerate}[label=(\alph*)]
\item $\theta_{<j}\mapsto \ell_j\{\theta^\star_j(\theta_{<j})\mid\theta_{<j}\}$ is continuous at $\phi^\star_{<j}$.

\item For some $\delta>0$, $\theta_j\mapsto \ell_j(\theta_j\mid\theta_{<j})$ is continuous at $\theta^\star_j(\theta_{<j})$ for all $\theta_{<j}\in N_{<j,\delta}$.

\item For every $\epsilon>0$ there exists $\delta>0$ such that
\begin{align*}
	\liminf_n\inf_{\theta_{<j}\in N_{<j,\delta}}\inf_{\theta_j\in N_{j,\epsilon}^c} \left[\ell^{(n)}_j(\theta_j\mid\theta_{<j})-\ell_j\{\theta^\star_j(\theta_{<j})\mid\theta_{<j}\}\right] &> 0.
\end{align*}
\end{enumerate}
\end{assumption}

A sufficient condition for parts (a) and (b) of \Cref{as:concentration} to hold is that $(\theta_{<j},\theta_j)\mapsto\ell_j(\theta_j\mid \theta_{<j})$ and $\theta_{<j}\mapsto\theta^\star_j(\theta_{<j})$ are continuous. A subtle but important difference between part (c) of \Cref{as:concentration} and its analogue in \cite[Theorem 3]{miller2021asymptotic} is that the former includes an infimum over the conditional parameters $\theta_{<j}$. This ensures the loss minimizer $\theta^\star_j(\theta_{<j})$ is uniformly well separated for all $\theta_{<j}$ in a neighborhood $N_{<j,\delta}$ of $\phi^\star_{<j}$. As a consequence, $\pi_j$ concentrates around $\theta_j^\star(\theta_{<j})$ uniformly over $N_{<j,\delta}$.

\begin{proof}[of \Cref{thrm:concentration}]
We show $\Pi^{(n)}(N_\epsilon^c)\to 0$ by induction on $J$. The case $J=1$ is precisely \cite[Theorem 3]{miller2021asymptotic}. Fix $J>1$ and assume the result holds for the sequential Gibbs posterior $\Pi^{(n)}_{<J}$ associated to the losses $\ell^{(n)}_j$ for $j\in[J-1]$. Fixing $\epsilon>0$,
\begin{align*}
	N_\epsilon^c &= \bigg\{\theta : \sum_{j=1}^J d_j^2(\theta_j,\phi^\star_j)\geq\epsilon^2\bigg\}
		\subseteq \left\{\theta : d_{<J}(\theta_{<J},\phi^\star_{<J})\geq \epsilon'\right\}\cup\left\{\theta:d_J(\theta_J,\phi^\star_J)\geq\epsilon'\right\} \\
			&= \left(N_{<J,\epsilon'}^c\times\mathcal{M}_J\right)\cup \left(\mathcal{M}_{<J}\times N_{J,\epsilon'}^c\right),
\end{align*}
where $\epsilon'=\epsilon/\sqrt{2}$. Therefore, since $\Pi^{(n)}_{<J}$ is the marginal of $\Pi^{(n)}$ over $\mathcal{M}_J$,
\begin{align}\label{eq:union_bound}
	\Pi^{(n)}(N_\epsilon^c) &\leq \Pi^{(n)}_{<J}(N_{<J,\epsilon'}^c) + \Pi^{(n)}(\mathcal{M}_{<J}\times N_{J,\epsilon'}^c).
\end{align}
$\Pi^{(n)}_{<J}(N_{<J,\epsilon'}^c)\to 0$ by inductive hypothesis; it remains to show $\Pi^{(n)}(\mathcal{M}_{<J}\times N_{J,\epsilon'}^c)\to 0$. Note
\begin{enumerate}[label=(\roman*)]
\item By \Cref{as:concentration}(c) there exist $\beta>0$ and $\delta_1>0$ such that for all $n$ sufficiently large,
\begin{align*}
	\inf_{\theta_{<J}\in N_{<J,\delta_1}}\inf_{\theta_J\in N_{J,\epsilon'}^c} \left[\ell^n_J(\theta_J\mid\theta_{<J})-\ell_J\{\theta^\star_J(\theta_{<J})\mid\theta_{<J}\}\right] &\geq 3\beta 
		> 0.
\end{align*}

\item By parts (a) and (b) of \Cref{as:concentration} there exists $\delta_2>0$ such that $\lvert\ell_J\{\theta^\star_J(\theta_{<J})\mid\theta_{<J}\} - \ell_J(\phi^\star_J\mid\phi^\star_{<J}\}\rvert<\beta$ and $\lvert\ell_J(\theta_J\mid\theta_{<J})-\ell_J(\theta^\star_J(\theta_{<J})\mid\theta_{<J})\rvert<\beta/2$ for all $\theta_{<J}\in N_{<J,\delta_2}$.
\end{enumerate}
Set $\delta=\min\{\delta_1,\delta_2\}$. We have
\begin{align*}
	\Pi^{(n)}(\mathcal{M}_{<J}\times N_{J,\epsilon'}^c) &= \int_{N_{<J,\delta}^c}\int_{N_{J,\epsilon'}^c}\Pi^{(n)}_J(d\theta_J\mid\theta_{<J}) \Pi^{(n)}_{<J}(d\theta_{<J}) \\
		&\qquad +\int_{N_{<J,\delta}}\int_{N_{J,\epsilon'}^c}\Pi^{(n)}_J(d\theta_J\mid\theta_{<J}) \Pi^{(n)}_{<J}(d\theta_{<J}).
\end{align*}
The inductive hypothesis together with $\int_{N_{J,\epsilon'}^c}\Pi^{(n)}_J(d\theta_J\mid\theta_{<J})\leq 1$ for all $\theta_{<J}$ imply
\begin{align*}
	\int_{N_{<J,\delta}^c}\int_{N_{J,\epsilon'}^c}\Pi^{(n)}_J(d\theta_J\mid\theta_{<J}) \Pi^{(n)}_{<J}(d\theta_{<J})d\theta_{<J} &\leq \Pi^{(n)}_{<J}(N_{<J,\delta}^c)
		\to 0.
\end{align*}
Define $f_n(\theta_{<J})=\int_{N_{J,\epsilon'}^c}\Pi^{(n)}_J(d\theta_J\mid\theta_{<J})$. If $f_n(\theta_{<J})\to 0$ uniformly on $N_{<J,\delta}$, then for any $\gamma>0$ we have $f_n(\theta_{<J})<\gamma$ for all $\theta_{<J}\in N_{<J,\delta}$ and all $n$ sufficiently large and hence 
\begin{align*}
	\int_{N_{<J,\delta}}\int_{N_{J,\epsilon'}^c}\Pi^{(n)}_J(d\theta_J\mid\theta_{<J}) \Pi^{(n)}_{<J}(d\theta_{<J}) &= \int_{N_{<J,\delta}}f_n(\theta_{<j})\Pi^{(n)}_{<J}(d\theta_{<J})
		\leq \gamma.
\end{align*}
Thus the integral vanishes and the proof is done. To verify $f_n(\theta_{<J})\to 0$ uniformly on $N_{<J,\delta}$,
\begin{equation}\label{eq:dominated}
\begin{aligned}
	f_n(\theta_{<J}) &= \frac{\int_{N_{J,\epsilon'}^c}\exp\{-n\ell^{(n)}_J(\theta_J\mid\theta_{<J})\}\Pi^{(0)}_J(d\theta_J)}{\int_{\mathcal{M}_J}\exp\{-n\ell^{(n)}_J(\theta_J\mid\theta_{<J})\}\Pi^{(0)}_J(d\theta_J)} \\
		&= \frac{\exp[n\{\ell_J(\phi^\star_J\mid\phi^\star_{<J})+2\beta\}]\int_{N_{J,\epsilon'}^c}\exp\{-n\ell^{(n)}_J(\theta_J\mid\theta_{<J})\}\Pi^{(0)}_J(d\theta_J)}{\exp[n\{\ell_J(\phi^\star_J\mid\phi^\star_{<J})+2\beta\}]\int_{\mathcal{M}_J}\exp\{-n\ell^{(n)}_J(\theta_J\mid\theta_{<J})\}\Pi^{(0)}_J(d\theta_J)}.
\end{aligned}
\end{equation}
By our choice of $\beta$ and $\delta$, for all $n$ sufficiently large and all $\theta_{<J}\in N_{<J,\delta}$ and $\theta_J\in N_{J,\epsilon'}^c$,
\begin{align*}
	\ell^n_J(\theta_J\mid\theta_{<J}) - \ell_J(\phi^\star_J\mid\phi^\star_{<J}) - 2\beta &= \ell^n_J(\theta_J\mid\theta_{<J}) - \ell_J(\theta^\star_J(\theta_{<J})\mid\theta_{<J}) \\
		&\qquad + \ell_J\{\theta^\star_J(\theta_{<J})\mid\theta_{<J}\} - \ell_J(\phi^\star_J\mid\phi^\star_{<J}) - 2\beta \\
		&\geq 3\beta - \beta - 2\beta = 0.
\end{align*}
So for all $n$ sufficiently large and all $\theta_{<J}\in N_{<J,\delta}$ the numerator in \eqref{eq:dominated} satisfies
\begin{align*}
	\exp[n\{\ell_J(\phi^\star_J\mid\phi^\star_{<J})+2\beta\}]\int_{N_{J,\epsilon'}^c}\exp\{-n\ell^{(n)}_J(\theta_J\mid\theta_{<J})\}\Pi^{(0)}_J(d\theta_J) &\leq 1.
\end{align*} 
Again by our choices of $\beta$ and $\delta$ and since $\ell^{(n)}_j(\cdot\mid\theta_{<J})\to\ell_j(\cdot\mid\theta_{<J})$ almost surely,
\begin{align*}
	\ell^n_J(\theta_J\mid\theta_{<J})-\ell_J(\phi^\star_J\mid\phi^\star_{<J})-2\beta &\to \ell_J(\theta_J\mid\theta_{<J})-\ell_J(\phi^\star_J\mid\phi^\star_{<J})-2\beta \\
		&= \ell_J(\theta_J\mid\theta_{<J}) - \ell_J\{\theta^\star_J(\theta_{<J})\mid\theta_{<J}\} \\
		&\qquad + \ell_J\{\theta^\star_J(\theta_{<J})\mid\theta_{<J}\} - \ell_J(\phi^\star_J\mid\phi^\star_{<J})-2\beta \\
		&\leq \tfrac{\beta}{2} + \beta - 2\beta
		< 0.
\end{align*}
So $\exp[-n\{\ell^n_J(\theta_J\mid\theta_{<J})-\ell_J(\phi^\star_J\mid\phi^\star_{<J})-2\beta\}]\to \infty$ and, since $\Pi^{(0)}(N_{j,\epsilon})>0$ for all $\epsilon>0$, the denominator in \eqref{eq:dominated} goes to $\infty$ by Fatou's lemma. Since the numerator in \eqref{eq:dominated} is bounded for $n$ sufficiently large and the denominator goes to infinity, $f_n(\theta_{<J})\to 0$ uniformly on $N_{<J,\delta}$.
\end{proof}

\subsection{Proof of \Cref{thrm:sequential}}\label{sec:sequential_bvm_proof}

We begin with the additional assumptions for \Cref{thrm:sequential}.

\begin{assumption}[Uniformly bounded third derivatives]\label{as:third_seq}
For all $\theta_{<j}$ there is an open, bounded $E_j\subseteq\mathcal{M}_j$ and chart $(V_j,\psi_j)$ with $\phi_j^\star\in E_j\cap V_j$ such that $\ell_j^{(n)}(\cdot\mid\theta_{<j})$ has continuous third derivatives on $E_j$ and
\begin{align}\label{eq:bound_seq}
    \sup_n\sup_{\theta_j\in E_j\cap V_j}\sup_{a,b,c}\lvert\partial_{abc}\{\ell_j^{(n)}(\cdot\mid\theta_{<j})\circ\psi_j^{-1}\}\{\psi_j(\theta_j)\}\rvert &< \infty.
\end{align}
\end{assumption}

\begin{assumption}[Positive definite Hessians]\label{as:hess_seq}
$H_j=\ell_j''(\phi_j^\star\mid \phi_{<j}^\star)$ is positive definite.
\end{assumption}

\begin{assumption}[Well separated minimizers]\label{as:min_sep_seq}
Shrinking $K_j$ in \Cref{as:seq_min} if necessary (but keeping $K_j$ compact and $\phi_j^\star$ in its interior), we have $K_j \subseteq E_j$ and, for all $\tilde \theta\in (\tau_{<j}^{(n)}\circ\varphi_{<j})(U_{<j})$,
\begin{align*}
    \liminf_n\inf_{\theta\in U_j\setminus K_j}[\ell_j^{(n)}\{\theta\mid (\tau_{<j}^{(n)}\circ\varphi_{<j})^{-1}(\tilde \theta)\}- \ell_j(\phi_j^\star\mid\phi_{<j}^\star)]>0.
\end{align*}
\end{assumption}

\begin{assumption}[Uniform convergence of minimizers]\label{as:min_seq}
    For any $\tilde \theta\in (\tau_{<j}^{(n)}\circ\varphi_{<j})(U_{<j})$, 
    \begin{align*}
        \theta_j^{(n)}\{(\tau_{<j}^{(n)}\circ\varphi_{<j})^{-1}(\tilde \theta)\}\to\phi_{j}^\star.
    \end{align*}
\end{assumption}

\begin{assumption}[Uniform convergence of losses]\label{as:loss_seq}
    For any $\tilde \theta\in (\tau_{<j}^{(n)}\circ\varphi_{<j})(U_{<j})$, 
    \begin{align*}
        \ell_j^{(n)}\{\cdot\mid (\tau_{<j}^{(n)}\circ\varphi_{<j})^{-1}(\tilde \theta)\}\to \ell_j(\cdot\mid \phi_{<j}^\star).
    \end{align*}
\end{assumption}
We clarify assumptions in the case $J=2$. One can sample from the sequential Gibbs posterior by sampling $\theta_1\sim \pi_1^{(n)}$ and $\theta_2\mid\theta_1\sim\pi_2^{(n)}(\cdot\mid\theta_1)$. Now let $(x_1,x_2)$ be the image of $(\theta_1,\theta_2)$ after mapping to Euclidean space with $\varphi$ and centering/scaling with $\tau^{(n)}$. We need to recover $\theta_1$ in order to specify the distribution of $x_2$, and this requires undoing $\tau_1$ and then $\varphi_1$. Explicitly, 
\begin{align*}
    \theta_1=(\tau_{1}^{(n)}\circ\varphi_{1})^{-1}(x_1) = \varphi_1^{-1}(x_1/\sqrt{n}+\tilde \theta_1^{(n)}).
\end{align*}
As $n\to\infty$, $x_1/\sqrt{n}\to 0$ and $\tilde \theta_1^{(n)}\to \tilde \phi_1^\star$. So $\theta_1\to \phi_1^\star$ and for large $n$, conditioning on $(\tau_{<j}^{(n)}\circ\varphi_{<j})^{-1}(\tilde \theta)$ is similar to conditioning on $\phi_{<j}^\star$. Assumptions \ref{as:third_seq}-\ref{as:loss_seq} are exactly those required to apply \Cref{thrm:bvm} to $\theta_2\mid\theta_1=\theta_2\mid \varphi_1^{-1}(x_1/\sqrt{n}+\tilde \theta_1^{(n)})$. The additional complexity is that the conditional parameters now vary with $n$, hence uniform convergence (\ref{as:min_seq}-\ref{as:loss_seq}) is required to evaluate limits such as $\lim_n\ell^{(n)}\{\cdot\mid \varphi_1^{-1}(x_1/\sqrt{n}+\tilde \theta_1^{(n)})\}$. The following is used in the proof of \Cref{thrm:sequential} to calculate $(\tau^{(n)}\circ\varphi)_\#\Pi^{(n)}$.

\begin{lemma}\label{lem:pushforward}
Let $X\subseteq\mathbb{R}^m$ and $Y\subseteq\mathbb{R}^n$ be open. Assume $\alpha:X\to \alpha(X)=U$ is a $C^1$-diffeomorphism and $\beta:X\times Y\to \mathbb{R}^n$ satisfies the following: For each $x\in X$, the map $\beta_x:Y\to \beta_x(Y)=V_x$ given by $\beta_x(y)=\beta(x,y)$ is a $C^1$-diffeomorphism. Define $f:X\times Y\to\mathbb{R}^m\times\mathbb{R}^n$ by $f(x,y)=\{\alpha(x),\beta_x(y)\}$. If $\Pi$ is a probability distribution on $X\times Y$ with density $\pi(x,y)=\pi_2(y\mid x)\pi_1(x)$, then $f_\#\Pi$ is a probability distribution on $f(X\times Y)$ with density
\begin{align*}
	f_\#\pi(u,v) &= (\beta_{\alpha^{-1}(u)})_\#\pi_2\{v\mid\alpha^{-1}(u)\}\alpha_\#\pi_1(u)
\end{align*}
where $f_\#\pi$ and $\alpha_\#\pi_1$ are the pushforward densities of $\pi$ and $\pi_1$ by $f$ and $\alpha$, respectively, and $(\beta_x)_\#\pi_2(v\mid x)$ is the pushforward of the conditional probability density $\pi_2(\cdot\mid x)$ on $Y$ by $\beta_x$.
\end{lemma}

\begin{proof}
Let $1_A$ denote the indicator function on a set $A$. The image of $X\times Y$ under $f$ satisfies
\begin{align*}
	f(X\times Y) &= \{(\alpha(x),\beta_x(y)) : x\in X\ \text{and}\ y\in Y\} \\
		&= \{(u,v) : \alpha^{-1}(u)\in X\ \text{and}\ \beta_{\alpha^{-1}(u)}^{-1}(v)\in Y\} \\
		&= \{(u,v) : u\in U\ \text{and}\ v\in V_{\alpha^{-1}(u)}\}.
\end{align*}
So $1_{f(X\times Y)}\{(u,v)\}=1_U(u)1_{V_{\alpha^{-1}(u)}}(v)$. For any measurable subset $A$ of $f(X\times Y)$,
\begin{align*}
	f_\#\Pi(A) 
	&= \int_Y\int_X 1_A\{\alpha(x),\beta_x(y)\}\pi_2(y\mid x)\pi_1(x)dxdy \\
	&= \int_Y\int_U 1_A\{u,\beta_{\alpha^{-1}(u)}(y)\}\pi_2\{y\mid\alpha^{-1}(u)\}\pi_1\{\alpha^{-1}(u)\}\lvert\det d\alpha^{-1}(u)\rvert dudy \\
	&= \int_U\int_Y 1_A\{u,\beta_{\alpha^{-1}(u)}(y)\}\pi_2\{y\mid\alpha^{-1}(u)\}dy\ \alpha_\#\pi_1(u)du \\
	&= \int_U\int_{V_{\alpha^{-1}(u)}} 1_A(u,v)\pi_2\{\beta^{-1}_{\alpha^{-1}(u)}(v)\mid\alpha^{-1}(u)\}\lvert\det d_v\beta^{-1}_{\alpha^{-1}(u)}(v)\rvert dv\ \alpha_\#\pi_1(u)du \\
	&= \int_A(\beta_{\alpha^{-1}(u)})_\#\pi_2\{v\mid\alpha^{-1}(u)\}\alpha_\#\pi_1(u)dvdu.
\end{align*}
The second and fourth equalities are obtained by substituting $u=\alpha(x)$ and $v=\beta_{\alpha^{-1}(u)}(y)$, respectively. The change of variables formula is valid in each case since $\alpha$ is a $C^1$-diffeomorphism of $X$ and $\beta_x$ is a $C^1$-diffeomorphism of $Y$ for each $x$. The last equality follows from $1_{f(X\times Y)}(u,v)=1_U(u)1_{V_{\alpha^{-1}(u)}}(v)$.
\end{proof}

\begin{proof}[of \Cref{thrm:sequential}]
We proceed by induction. When $J=1$, \Cref{thrm:bvm} provides
\begin{align*}
    (\tau_1^{(n)}\circ\varphi_1)_\#\Pi_1^{(n)}\to N(0, H_{\varphi_1}^{-1})
\end{align*}
in total variation, hence setwise. Fix $J>1$, let $\Pi^{(n)}$ be the sequential Gibbs posterior associated to $\{\ell^{(n)}\}_{j=1}^J$, and let $\pi^{(n)}$ be its density. Similarly, let $\Pi^{(n)}_{<J}$ be the sequential Gibbs posterior associated to $\{\ell^{(n)}_j\}_{j=1}^{J-1}$, and let $\pi^{(n)}_{<j}$ be its density. Then $\tilde\Pi^{(n)}_{<J}=(\tau^{(n)}_{<J}\circ\varphi_{<J})_\#\Pi^{(n)}_{<J}$ has density $\tilde\pi^{(n)}_{<J}=(\tau_{<J}^{(n)}\circ\varphi_{<J})_\#\pi_{<J}^{(n)}$. Assume by inductive hypothesis that $\tilde\Pi^{(n)}_{<J}\to \prod_{j=1}^{J-1} N(0,H_{\varphi_j}^{-1})$ setwise, and let $\tau^{(n)}$ and $\varphi$ be as in the statement of the theorem. By \Cref{lem:pushforward},
    \begin{align*}
        (\tau^{(n)}\circ\varphi)_\#\pi^{(n)}(\tilde\theta) &= \prod_{j=1}^J (\tau^{(n)}_j\circ\varphi_j)_\#\pi^{(n)}_j\{\tilde\theta_j\mid(\tau^{(n)}_{<j}\circ\varphi_{<j})^{-1}(\tilde\theta_{<j})\} \\
         &= (\tau^{(n)}_J\circ\varphi_J)_\#\pi^{(n)}_J\{\tilde\theta_J\mid(\tau^{(n)}_{<J}\circ\varphi_{<J})^{-1}(\tilde\theta_{<J})\}\tilde\pi^{(n)}_{<J}(\tilde\theta_{<J}),
    \end{align*}
where each $\pi^{(n)}_j$ is the density of $\Pi^{(n)}_j$. As discussed prior to the statement of \Cref{lem:pushforward}, Assumptions \ref{as:third_seq}-\ref{as:loss_seq} suffice to apply \Cref{thrm:bvm} to $\Pi^{(n)}_J$, collectively implying that
    \begin{align*}
        (\tau_{J}^{(n)}\circ\varphi_{J})_\#\Pi_J^{(n)}\{(\tau_{<J}^{(n)}\circ\varphi_{<J})^{-1}(\tilde \theta_{<J})\} \to N(0, H_{\varphi_J}^{-1})
    \end{align*}
    in total variation -- and hence setwise -- for every $\tilde \theta_{<J}\in (\tau^{(n)}_{<J}\circ \varphi_{<J})(U_1\times\cdots\times U_{J-1})$. It then immediately follows by Theorem 1 in \cite{sethuraman1961some} and the inductive hypothesis that
\begin{align*}
	(\tau^{(n)}\circ\varphi)_\#\Pi^{(n)}\to \prod_{j=1}^{J}N(0, H_{\varphi_j}^{-1})
\end{align*}
setwise, completing the proof.
\end{proof}


\subsection{Proof of \Cref{prop:pca_bvm}}\label{sec:prop1}


We now apply \Cref{thrm:sequential} to principal component analysis, resulting in \Cref{prop:pca_bvm}. 

\begin{proof}[of \Cref{prop:pca_bvm}]
Let $\hat\lambda_1>\cdots>\hat\lambda_p>0$ be the eigenvalues of the empirical covariance and set $\Lambda^{(n)}=\text{diag}(\hat\lambda_1,...,\hat\lambda_p)$. The $j$th finite sample loss is
    \begin{align*}
        \ell_j^{(n)}(w_j\mid v_{<j}) &= -w_j^TN_{<j}^T\hat\Lambda^{(n)} N_{<j}w_j, 
    \end{align*}
where $N_{<j}$ is a basis for the null space of $v_{<j}$. By the strong law of large numbers this converges almost surely to
    \begin{align*}
        \ell_j(w_j\mid v_{<j}) &= -w_j^TN_{<j}^T\Lambda N_{<j}w_j, 
    \end{align*}
where $\Lambda=\text{diag}(\lambda_1,...,\lambda_p)$ and $\lambda_1>\cdots>\lambda_p>0$ are the eigenvalues of $var(x)$. Thus \Cref{as:converge} holds. \Cref{as:min} follows from the fact that quadratic forms are differentiable and maximized by the leading eigenvector. In particular,  $w_j^{(n)}(v_{<j})$ is the leading eigenvector of $N_{<j}^T\hat\Lambda^{(n)} N_{<j}$ and $w_j^{\star}(v_{<j})$ is the leading eigenvector of $N_{<j}^T\Lambda_{\geq j} N_{<j}$. The sequential minimizers are $\phi_j^\star=(1,0,...,0)\in\mathbb{S}^{p-k}$, and $N_{<j}^\star=[e_{j},...e_p]$ after conditioning on $\phi_1^\star,...,\phi_{j-1}^\star$. Explicitly,
\begin{align*}
    \ell_j(w_j\mid \phi_{<j}^\star) &= w_j^T\Lambda_{\geq j}w_j
\end{align*}
with $\Lambda_{>j}=\text{diag}(\lambda_{j},...,\lambda_p)$.

The rest of the proof requires charts. Define $U_j=\{w\in\mathbb{S}^{p-j}\mid w_1>0\}$ and $\varphi_j(w)=w_{-1}$ where $w_{-1}\in\mathbb{R}^{p-j}$ is the vector obtained by deleting the first entry of $w$. The inverse $\psi_j=\varphi_j^{-1}$ maps $u=(u_1,...,u_{p-j})$ to $\psi_j(u)=(\sqrt{1-||u||^2}, u_1, ..., u_{p-j})$. We always have $\phi_j^\star\in U_j$, and a uniform prior is always positive and continuous at $\phi_j^\star$. 

We compute third derivatives for \Cref{as:third_seq}. Substituting $w_j=\psi_j(u)$ into $\ell_j^{(n)}(w_j\mid v_{<j})$ and applying the product rule:
{\footnotesize 
\begin{align*}
    \partial_{abc}\{\psi_j(u)^TN_{<j}^T\Lambda^{(n)}N_{<j}\psi_j(u)\} &= 2\{\psi_j(u)^TN_{<j}^T\Lambda^{(n)}N_{<j}\partial_{abc}\psi_j(u) + \partial_a\psi_j(u)^TN_{<j}^T\Lambda^{(n)}N_{<j}\partial_{bc}\psi_j(u) \\
		&\ + \partial_b\psi_j(u)^TN_{<j}^T\Lambda^{(n)}N_{<j}\partial_{ac}\psi_j(u) +\partial_c\psi_j(u)^TN_{<j}^T\Lambda^{(n)}N_{<j}\partial_{ab}\psi_j(u)\}
\end{align*}}%
\noindent Recall $x^TAy \leq \lambda_\text{max}(A)||x||||y||$, where $\lambda_\text{max}$ is the largest eigenvalue of $A$. Applying the triangle inequality and using the bound $\lambda_\text{max}(A)(N_{<j}^T\Lambda^{(n)}N_{<j}) \leq \hat\lambda_1$,
{\small
\begin{align*}
    ||\partial_{abc}\{\psi_j(u)^TN_{<j}^T\Lambda^{(n)}N_{<j}\psi_j(u)\}|| &\leq 2\hat\lambda_1\{||\psi_j(u)|| ||\partial_{abc}\psi_j(u)|| + ||\partial_a\psi_j(u)|| ||\partial_{bc}\psi_j(u)|| + \\&\qquad\qquad||\partial_b\psi_j(u)|| ||\partial_{ac}\psi_j(u)|| + ||\partial_c\psi_j(u)|| ||\partial_{ab}\psi_j(u)||\}.
\end{align*}}%
The map $\psi$ is smooth on the compact set $B_{1/2}(0)=\{u\mid ||u||\leq 1/2\}\}\subseteq \varphi_j(U_j)$. Hence, there is a constant $C$ such that
{\footnotesize
\begin{align*}
    \max\left\{\sup_{u\in B_{1/2}(0)} ||\psi_j(u)||, \sup_{u\in B_{1/2}(0)}\sup_{a}||\partial_a\psi_j(u)||, \sup_{u\in B_{1/2}(0)}\sup_{a, b} ||\partial_{ab}\psi_j(u)||, \sup_{u\in B_{1/2}(0)}\sup_{a, b, c}||\partial_{abc}\psi_j(u)||\right\} \leq C
\end{align*}}%
where $\sup_a$, $\sup_{a,b}$, and $\sup_{a,b,c}$ range over all possible first, second, and third derivatives, respectively. Therefore
\begin{align*}
    \sup_{u\in B_{1/2}(0)}\sup_{a, b, c}||\partial_{abc}\{\psi_j(u)^TN_{<j}^T\Lambda^{(n)}N_{<j}\psi_j(u)\}|| &\leq 8C\hat\lambda_1
\end{align*}
which is bounded in the limit because $\hat \lambda_1\to \lambda_1$. Hence, \Cref{as:third_seq} holds.

We compute the Hessian for \Cref{as:hess_seq}. Substituting $w_j=\psi_j(u)$ into $\ell_j(w_j\mid \phi_{<j}^\star)$,
\begin{align*}
    \ell_j\{\psi_j(u) \mid \phi_{<j}^\star\} &= \psi_j(u)^T\Lambda_{\geq j}\psi_j(u) \\
    &-\lambda_j(1-||u||^2) - \sum_{i=j+1}^{p} \lambda_i u_i^2 \\
    &= -\lambda_j+\lambda_j\sum_{i=j+1}^{p}u_j^2 - \sum_{i=j+1}^{p} \lambda_i u_i^2 \\
    &= -\lambda_j + \sum_{i=j+1}^{p}(\lambda_j-\lambda_i)u_i^2.
\end{align*}
Therefore, $H_j = \text{diag}(\lambda_j-\lambda_{j+1},\ldots,\lambda_j-\lambda_{p})$, which is positive definite because the eigenvalues are distinct. This calculation is well-defined because $\phi_j^\star$ is a critical point of $\ell_j(\cdot \mid \phi_{<j}^\star)$. 

We temporarily postpone \Cref{as:min_sep_seq}. \Cref{as:min_seq} and \Cref{as:loss_seq} are proved simultaneously by induction. The $j=1$ case is automatic from $\Lambda^{(n)} \to \Lambda$ and elementary perturbation theory (for example, from \cite{magnus1985differentiating}). Now assume \Cref{as:min_seq} and \Cref{as:loss_seq} for $k\in [j-1]$. By \Cref{as:min_seq}, $(\tau_{<j}^{(n)}\circ\varphi_{<j})^{-1}(\tilde v)\}\to \phi_{<j}^\star$ for any $\tilde v\in (\tau_{<j}^{(n)}\circ\varphi_{<j})(U_{<j})$. Let $N_{<j}(\tilde v)$ be a basis for the null space of $(\tau_{<j}^{(n)}\circ\varphi_{<j})^{-1}(\tilde v)$. Without loss of generality we assume the map $\tilde v\mapsto N_{<j}(\tilde v)$ is continuous, so $N_{<j}(\tilde v)\to N_{<j}^\star=[e_j,...,e_p]$ as $(\tau_{<j}^{(n)}\circ\varphi_{<j})^{-1}(\tilde v)\}\to \phi_{<j}^\star$. This can be achieved, for example, using the Gram-Schmidt algorithm and the fact that both $\varphi_{<j}$ and $\tau_{<j}^{(n)}$ have continuous inverses. Therefore,
\begin{align*}
    N_{<j}(\tilde v)^T\Lambda^{(n)}N_{<j}(\tilde v) \to N_{<j}^{\star T}\Lambda N_{<j}^{\star}=\Lambda_{\geq j},
\end{align*}
and hence $\ell_j^{(n)}\{\cdot\mid (\tau_{<j}^{(n)}\circ\varphi_{<j})^{-1}(\tilde v)\}\to \ell_j(\cdot\mid \phi_{<j}^\star)$. Again by perturbation theory, $w_j^{(n)}\{(\tau_{<j}^{(n)}\circ\varphi_{<j})^{-1}(\tilde v)\}\to\phi_{j}^\star$, validating \Cref{as:min_seq} and \Cref{as:loss_seq}.

We now verify \Cref{as:min_sep_seq}. First, we show the convergence $\ell_j^{(n)}\{\cdot\mid (\tau_{<j}^{(n)}\circ\varphi_{<j})^{-1}(\tilde v)\}\to \ell_j(\cdot\mid \phi_{<j}^\star)$ can be upgraded from pointwise to uniform. We have
{\small
\begin{align*}
    \sup_{w_j}|\ell_j^{(n)}\{w_j\mid (\tau_{<j}^{(n)}\circ\varphi_{<j})^{-1}(\tilde v)\}- \ell_j(w_j\mid \phi_{<j}^\star)| &= \sup_{w_j}|w_{j}^TN_{<j}(\tilde v)^T\Lambda^{(n)}N_{<j}(\tilde v)w_j^T - w_j^T\Lambda_{\geq j}w_j| \\
    &= \lambda_{\text{max}}\{N_{<j}(\tilde v)^T\Lambda^{(n)}N_{<j}(\tilde v) -\Lambda_{\geq j}\}.
\end{align*}}%
This vanishes because $N_{<j}(\tilde v)^T\Lambda^{(n)}N_{<j}(\tilde v) \to \Lambda_{\geq j}$, hence the convergence is uniform. Uniform convergence allows us to swap the liminf and the infimum \citep{braides2006handbook} to obtain: 
\begin{align*}
    &\liminf_n\inf_{w_j\in U_j\setminus K_j}[\ell_j^{(n)}\{w_j\mid (\tau_{<j}^{(n)}\circ\varphi_{<j})^{-1}(\tilde v)\}- \ell_j(\phi_j^\star\mid\phi_{<j}^\star)] \\
    &\hphantom{-------}= \inf_{w_j\in U_j\setminus K_j}\liminf_n[\ell_j^{(n)}\{w_j\mid (\tau_{<j}^{(n)}\circ\varphi_{<j})^{-1}(\tilde v)\}- \ell_j(\phi_j^\star\mid\phi_{<j}^\star)] \\
    &\hphantom{-------}= \inf_{w_j\in U_j\setminus K_j}\{\ell_j(w_j\mid\phi_{<j}^\star)- \ell_j(\phi_j^\star\mid\phi_{<j}^\star)\}
\end{align*}
which is always strictly greater than zero for any compact $K_j$ with $\phi_j^\star$ in the interior because $\phi_j^\star$ is the global minimizer of $\ell_j(\cdot \mid\phi_{<j}^\star)$ over $U_j$. This proves the asymptotic product normal form in \Cref{prop:pca_bvm}. The concentration result follows from the fact that setwise convergence implies convergence in probability.
\end{proof}


\section{Geometry background}\label{sec:geometry}


In this section we show all derivative conditions in Theorems \ref{thrm:bvm}, \ref{thrm:concentration}, and \ref{thrm:sequential} are well-defined, meaning they do not depend on choice of chart. Thus these conditions may be verified by mapping $\ell^{(n)}$ and $\ell$ to Euclidean space via any one chart and taking usual partial derivatives. Let $\mathcal{M}$ be a smooth $p$-dimensional manifold as before. We will prove the following.

\begin{lemma}\label{lem:invariant}
Assume $\ell^{(n)}:\mathcal{M}\to\mathbb{R}$ converges almost surely to $\ell$ and fix $\phi^\star\in\mathcal{M}$.
\begin{enumerate}
\item If $\ell'(\phi^\star)=0$ then the Hessian $\ell''(\phi^\star)$ is well-defined.

\item If \Cref{as:third} holds for some chart $(V,\psi)$, then it holds for every chart $(U,\varphi)$ containing $\phi^\star$ with $U$ and $\varphi$ replacing $V$ and $\psi$ in \Cref{eq:bound}, respectively.

\item Suppose $E$ satisfies \Cref{as:third}. Let $(U,\varphi)$ be any chart such that $U\cap E\neq\emptyset$ and set $\tilde\ell^{(n)}=\ell^{(n)}\circ\varphi^{-1}$ and $\tilde\ell=\ell\circ\varphi^{-1}$. Then there exists an open $E'\subseteq E$ containing $\phi^\star$ such that $\tilde\ell'$ and $\tilde\ell''$ exist on $\varphi(E')$, and $\tilde\ell^{(n)}$, $\tilde\ell^{(n)'}$, and $\tilde\ell^{(n)''}$ converge uniformly on $\varphi(E')$ to $\tilde\ell$, $\tilde\ell'$, and $\tilde\ell''$, respectively. In particular, if $\theta^{(n)}\to\phi^\star$ with $\ell^{(n)'}(\theta^{(n)})=0$ for all $n$, then $\ell'(\phi^\star)=0$ and the Hessian $\ell''(\phi^\star)$ is well-defined.

\end{enumerate}
\end{lemma}

We begin by defining first derivatives on manifolds as in \cite{lee2012smooth}. A function $f:\mathcal{M}\to \mathbb{R}$ is smooth at $\theta\in \mathcal{M}$ if for all charts $(U,\varphi)$ with $\theta\in U$, there is an open neighborhood $U_\theta\subseteq U$ such that $f\circ\varphi^{-1}$ is smooth on $\varphi(U_\theta)$. Let $\mathcal{C}^\infty(\mathcal{M})$ be the set of smooth functions from $\mathcal{M}$ to $\mathbb{R}$. Fix $f\in\mathcal{C}^\infty(\mathcal{M})$, a chart $(U,\varphi)$, and a point $\theta\in U$. The $i$th partial derivative of $f$ at $\theta$ is
\begin{align*}
    \partial_i\vert_\theta f &= \partial_i f(\theta) 
    = \partial_i (f\circ\varphi^{-1})[\varphi(\theta)]
\end{align*}
where on the right side $\partial_i$ is the usual partial derivative of $f\circ\varphi^{-1}:\varphi(U)\to\mathbb{R}$. Abbreviating $\partial_i\vert_\theta$ to $\partial_i$ when $\theta$ is arbitrary or understood, it follows from standard rules of differentiation that $\partial_i$ is a linear operator on $\mathcal{C}^\infty(\mathcal{M})$ and satisfies the product rule $\partial_i(fg)=f\partial_ig + g\partial_i f$ for all $f,g\in\mathcal{C}^\infty(\mathcal{M})$. More generally, a derivation at $\theta\in\mathcal{M}$ is a linear map $v:\mathcal{C}^\infty(\mathcal{M})\to\mathbb{R}$ satisfying $v(fg)=f(\theta)vg + g(\theta)vf$ for all $f,g\in\mathcal{C}^\infty(\mathcal{M})$. The tangent space to $\theta$ at $\mathcal{M}$ is the collection $T_\theta\mathcal{M}$ of derivations at $\theta$. $T_\theta\mathcal{M}$ is a $p$-dimensional vector space and $\{\partial_i\}_{i=1}^p$ is a basis for $T_\theta\mathcal{M}$ that depends on the chart and $\theta$. The differential of $f\in\mathcal{C}^\infty(\mathcal{M})$ at $\theta$ is the linear map $f'(\theta):T_\theta\mathcal{M}\to T_{f(\theta)}\mathbb{R}$ defined by 
\begin{align}\label{eq:differential}
	f'(\theta)v &=v^i\partial_i f(\theta)
		= v^i\partial_i(f\circ\varphi^{-1})\{\varphi(\theta)\}
\end{align}
where $v=v^i\partial_i$ is the representation of $v\in T_\theta\mathcal{M}$ with respect to the basis $\{\partial_i\}$. Here and throughout we use Einstein notation with matching upper and lower indices understood as a sum. For example, $v^i\partial_i= v^1\partial_1+\cdots+v^p\partial_p$. An important property of $f'$ is that it is chart-invariant. To see this, let $\{\partial_i\}$ and $\{\hat\partial_j\}$ be bases for $T_\theta\mathcal{M}$ corresponding to charts $(U,\varphi)$ and $(V,\psi)$, respectively, and let $v=v^i\partial_i=\hat v^j\hat\partial_j$. A change of basis argument via the chain rule shows $\partial_i=\partial_i(\psi\circ\varphi^{-1})^j\hat\partial_j$ and hence $\hat v_j=v^i\partial_i(\psi\circ\varphi^{-1})^j$. Suppressing $\theta$, the chain rule gives
\begin{align*}
	f'(v^i\partial_i) &= v^i\partial_i(f\circ\varphi^{-1})
		= v^i\partial_i(f\circ\psi^{-1}\circ\psi\circ\varphi^{-1}) \\
		&= v^i\hat\partial_j(f\circ\psi^{-1})\partial_i(\psi\circ\varphi^{-1})^j
		= v^i\partial_i(\psi\circ\varphi^{-1})^jf'(\hat\partial_j) \\
		&= f'(\hat v_j\hat\partial_j).
\end{align*}
Thus the value of $f'(\theta)v$ at any $\theta$ is independent of chart. In particular, $f'(\theta)=0$ if and only if $(f\circ \varphi^{-1})'[\varphi(\theta)]=0$ for any chart containing $\theta$. Hence $\ell^{(n)'}(\theta_n)=0$ is well-defined.

This approach fails to provide a chart-invariant notion of second derivatives. In the Euclidean case, the Hessian of $\tilde f:\mathbb{R}^p\to\mathbb{R}$ at $x$ is the $p$-by-$p$ matrix $\tilde f''(x)$ whose $ij$th entry is $\partial_j\partial_i \tilde f(x)$. This defines a bilinear operator $\tilde f''(x):\mathbb{R}^p\times\mathbb{R}^p\to\mathbb{R}$ via $\tilde f''(x)(u,v)=u^iv^j\partial_j\partial_i\tilde f(x)$. Motivated by this, it is natural to try to define the Hessian of $f\in\mathcal{C}^\infty(\mathcal{M})$ at $\theta\in\mathcal{M}$ to be the bilinear operator $f''(\theta):T_\theta\mathcal{M}\times T_\theta\mathcal{M}\to\mathbb{R}$ given by
\begin{align}\label{eq:hessian}
	f''(\theta)(u,v) &= u^iv^j\partial_j\partial_i f(\theta)
\end{align}
where $\partial_j\partial_if=\partial_j(\partial_i f\circ\varphi^{-1})$ is computed by composing the previous notion of first derivatives. However, this is not chart-invariant in general. To see why, let $u=u^i\partial_i=\hat u^j\hat\partial_j$ and $v=v^i\partial_i=\hat v^j\hat\partial_j$ with $\{\partial_i\}$ and $\{\hat\partial_i\}$ corresponding to $(U,\varphi)$ and $(V,\psi)$, respectively. Setting $g=\psi\circ\varphi^{-1}$, repeated use of $\partial_i=\partial_ig^j\hat\partial_j$ and $\hat v^j=v^i\partial_ig^j$ yields
\begin{align*}
	f''(u^i\partial_i,v^j\partial_j) &= u^iv^j\partial_j\partial_i f \\
		&= u^iv^j\partial_j(\partial_ig^\alpha\hat\partial_\alpha f) \\
		&= u^iv^j(\partial_j\partial_ig^\alpha\hat\partial_\alpha f + \partial_ig^\alpha\partial_j\hat\partial_\alpha f) \\
		&= u^iv^j\partial_j\partial_ig^\alpha\hat\partial_\alpha f + u^iv^j\partial_ig^\alpha\partial_jg^\beta\hat\partial_\beta\hat\partial_\alpha f \\
		&= u^iv^j\partial_j\partial_ig^\alpha\hat\partial_\alpha f + \hat u^\alpha\hat v^\beta\hat\partial_\beta\hat\partial_\alpha f \\
		&= u^iv^j\partial_j\partial_ig^\alpha\hat\partial_\alpha f + f''(\hat u^\alpha\hat\partial_\alpha,\hat v^\beta\hat\partial_\beta).
\end{align*}
So $f''(u^i\partial_i,v^j\partial_j)$ and $f''(\hat u^\alpha\hat\partial_\alpha,\hat v^\beta\hat\partial_\beta)$ differ by $u^iv^j\partial_j\partial_ig^\alpha\hat\partial_\alpha f$ which is not zero in general. However, if $\theta$ is a critical point of $f$ then this term does vanish, leaving $f''(\theta)(u^i\partial_i,v^j\partial_j)=f''(\theta)(\hat u^\alpha\hat\partial_\alpha,\hat v^\beta\hat\partial_\beta)$. Thus \eqref{eq:hessian} is chart-invariant -- and hence well-defined -- precisely at critical points of $f$, proving \Cref{lem:invariant}.1. In particular, if $f'(\theta)=0$, one can check that $f''(\theta)$ is positive-definite by computing the Euclidean Hessian $(f\circ\varphi^{-1})''[\varphi(\theta)]$ in any chart $(U,\varphi)$ containing $\theta$.

Third derivatives are also not well-defined for analogous reasons. For simplicity, we avoid interpreting third derivatives as operators and instead define the tensor of third partial derivatives $f'''$ at $\theta\in U$ as the $p\times p\times p$ array with $ijk$th entry given by $\partial_k\partial_j\partial_i f(\theta)=\partial_k\partial_j\partial_i(f\circ\varphi^{-1})[\varphi(\theta)]$. We say a sequence $(f_n''')$ is uniformly bounded on $E\subseteq\mathcal{M}$ if there is at least one chart $(V,\psi)$ with corresponding partials $\partial_i$ and $V\cap E\neq \varnothing$ such that
\begin{align}\label{eq:bound1}
	\sup_n\sup_{\theta\in E\cap V}\sup_{i,j,k} \left\lvert \partial_k\partial_j\partial_i f_n(\theta)\right\rvert < \infty.
\end{align}
To prove \Cref{lem:invariant}.2, let $(V,\psi)$ and $E$ be as in \Cref{as:third} and let $(U,\varphi)$ be another chart with partials $\hat\partial_i$ and $U\cap E\neq \varnothing$. Letting $g=\psi\circ\varphi^{-1}$ we have
\begin{align*}
    \hat\partial_k\hat\partial_j\hat\partial_i f_n &= \hat\partial_k\hat\partial_j(\partial_\alpha f_n\hat\partial_i g^\alpha) \\
		&= \hat\partial_k(\partial_\beta\partial_\alpha f_n\hat\partial_j g^\beta\hat\partial_i g^\alpha + \partial_\alpha f_n\hat\partial_j\hat\partial_i g^\alpha) \\
		&= \partial_\gamma\partial_\beta\partial_\alpha f_n\hat\partial_k g^\gamma\hat\partial_jg^\beta\hat\partial_ig^\alpha \\
		&\qquad + \partial_\beta\partial_\alpha f_n(\hat\partial_ig^\alpha\hat\partial_k\hat\partial_jg^\beta + \hat\partial_jg^\beta\hat\partial_k\hat\partial_ig^\alpha + \hat\partial_kg^\beta\hat\partial_j\hat\partial_ig^\alpha) \\
		&\qquad\quad + \partial_\alpha f_n\hat\partial_k\hat\partial_j\hat\partial_ig^\alpha.
\end{align*}
Assume without loss of generality $E$ is properly contained in $U\cap V$. Then since $g=\psi\circ\varphi^{-1}:\varphi(U\cap V)\to \psi(U\cap V)$ is smooth on $U\cap V$ and $E\cap V\cap U$ is bounded, $g$ and all its partial derivatives up to and including order three are uniformly bounded on $\varphi(E\cap U\cap V)$. Hence
\begin{align*}
	\lvert \hat\partial_k\hat\partial_j\hat\partial_i f_n\rvert &\leq C\left(\sum_{i,j,k}\lvert\partial_k\partial_j\partial_i f_n\rvert + \sum_{i,j}\lvert\partial_j\partial_i f_n\rvert + \sum_i\lvert\partial_i f_n\rvert\right)
\end{align*}
for some finite constant $C$. Combining this with \eqref{eq:bound1} gives
\begin{align}\label{eq:bound2}
	\sup_n\sup_{\theta\in E\cap U\cap V}\sup_{i,j,k} \left\lvert \hat\partial_k\hat\partial_j\hat\partial_i f_n(\theta)\right\rvert < \infty.
\end{align}
Therefore if \eqref{eq:bound1} holds for one chart satisfying the assumptions of $(V,\psi)$, then a similar bound holds for any chart overlapping with $E$. This proves \Cref{lem:invariant}.2. 

Finally, \Cref{lem:invariant}.3 follows immediately from \Cref{lem:invariant}.1, \Cref{lem:invariant}.2, and \cite[Theorem 7]{miller2021asymptotic}, where, for a given chart $(U,\varphi)$, the set $E'$ is any open subset of $E\cap U$ such that $\varphi(E')$ is a convex subset of $\mathbb{R}^p$.

\newpage


\section{Applied details}\label{sec:simulations}


\subsection{Mean and variance simulations}

\Cref{table:coverage} shows the coverage of intervals for $\mu$ for the joint/sequential Gibbs posteriors when $X$ follows different distributions, including N$(0,1)$, $\text{t}_5(0,1)$, Skew-Normal$(0,1,1)$, and Gumbel$(0,1)$. A total of $500$ datasets were generated, each with $n=1000$ independent samples from one of the above distributions. Credible intervals were estimated for each dataset, and coverage was calculated as the proportion of credible intervals containing the truth. A grid search was performed to find calibration hyperparameters which yielded $95\%$ coverage. All other simulations and data analyses in this paper selected hyperparameters using the proposed calibration algorithm in Section 3 of the main text and the stochastic approximation method in Section 3.2 of the supplement.

The sequential Gibbs posterior was sampled exactly using the sequential decomposition. The joint Gibbs posterior was sampled with Metropolis-Hastings using the proposals $\mu\sim N(\mu^{(s-1)},\varepsilon^2)$ and $\sigma^2\sim N_+(\mu^{(s-1)},\varepsilon^2)$; $\varepsilon$ was tuned adaptively so the acceptance ratio was between $25\%$ and $50\%$.

\setcounter{algocf}{1} 
\begin{algorithm}[t]
\caption{Automatic hyperparameter selection}\label{alg:stochastic}
\KwData{Target radius $\hat r_b$, initial parameter $\eta_0$, step size $\varepsilon_t$.}
\KwResult{Hyperparameter $\eta$ with $\hat r_g(\eta)\approx \hat r_b$}
$t\gets 0$\;
\While{\text{not converged}}{
    $\delta_t \gets \{\hat r_g(\eta_t)-\hat r_b\}/\hat r_b$;\\
    $\eta_{t+1} \gets \eta_t \exp(\delta_t/\varepsilon_t)$; \\
    $t \gets t+1$;\\
}
\Return $\eta_t$
\end{algorithm}

\subsection{Sampling from the Sequential Posterior}\label{sec:sampling}
\textcolor{black}{In this subsection we discuss challenges associated with sampling from the sequential Gibbs posterior and outline solutions for future applications. Importantly, the algorithms presented may only be required for small sample sizes -- when $n$ is large, the sequential posterior can be well approximated by a Gaussian with analytic moments defined in Theorem 3. For readability, we simplify to the case when $J=2$ and focus on sampling $(\theta_1, \theta_2)$ according to the density
\begin{align*}
    \pi(\theta_1,\theta_2\mid x) &= \frac{1}{z_1(x)}\exp\{-\ell_1^{(n)}(\theta_1\mid x)\}\frac{1}{z_2(\theta_1, x)}\exp\{-\ell_2^{(n)}(\theta_2\mid \theta_1, x)\}\pi_1(\theta_1)\pi_2(\theta_2), \\
    z_1(x) &= \int_{\mathcal{M}_1}\exp\{-\ell_1^{(n)}(\theta_1\mid x)\}\pi_1(\theta_1)d\theta_1, \\ 
    z_2(\theta_1, x) &= \int_{\mathcal{M}_2}\exp\{-\ell_2^{(n)}(\theta_2\mid \theta_1, x)\}\pi_2(\theta_2)d\theta_2,
\end{align*}
where the calibration constants $\eta=(\eta_1, \eta_2)$ are omitted. If exact sampling is possible, such as in principal component analysis, then samples from the joint distribution may be obtained by drawing $\theta_1 \sim \pi_1(\cdot\mid x)$ and $\theta_2\mid \theta_1\sim \pi_2(\cdot\mid \theta_1, x)$.}

\textcolor{black}{In most situations, efficient exact sampling is not available, and the unknown normalizing constant $z_2(\theta_1, x)$ is problematic. Consider a standard Metropolis-Hastings algorithm, currently in state $(\theta_1^t,\theta_2^t)$. A new state $\theta_1'$ is proposed according to $q_1(\cdot \mid \theta_1^t,\theta_2^t)$. The probability of accepting $\theta_1'$ depends on the posterior ratio
\begin{align*}
    L(\theta_1'\mid \theta_1^t,\theta_2^t)&= \frac{\pi(\theta_1',\theta_2^{t}\mid x)}{\pi(\theta_1^{t},\theta_2^{t}\mid x)} = \frac{\exp\{-\ell_1^{(n)}(\theta_1'\mid x)\}}{\exp\{-\ell_1^{(n)}(\theta_1\mid x)\}}\frac{z_2(\theta_1^t, x)}{z_2(\theta_1', x)}\frac{\exp\{-\ell_2^{(n)}(\theta_2^{t}\mid \theta_1', x)\}}{\exp\{-\ell_2^{(n)}(\theta_2^t\mid \theta_1^t, x)\}}\frac{\pi_1(\theta_1')}{\pi_1(\theta_1^t)}.
\end{align*}
and the ratio of proposal densities. Importantly, the normalizing constants do not cancel. Although the sequential Gibbs posterior is new, identical sampling problems appear when studying cut posteriors in modular Bayesian analysis and several solutions have been proposed. One solution is to estimate the normalizing constants, for example using importance sampling. Many numerical integration methods have poor convergence properties for problems with high-dimensional parameters and may not be computationally feasible on manifolds.} 

\textcolor{black}{As a general solution, we instead prefer to construct an algorithm in which $\theta_1$ is sampled without evaluating $\pi_2(\cdot\mid x, \theta_1)$. A simple solution is to generate high-quality samples of $\theta_1\sim \pi_1(\cdot\mid x)$ and then treat these as exact samples when drawing $\theta_2\mid \theta_1$. For example, in the first stage one could construct an MCMC algorithm for $\theta_1$ targeting $\pi_1(\cdot\mid x)$ and run this algorithm until convergence to obtain samples $\theta_1^t$, $t=1,\dots,T$. In the second stage one constructs an MCMC algorithm targeting $\pi_2(\cdot\mid\theta_1^t)$ for each $t$; this is run until convergence to obtain samples $\theta_2^s\mid \theta_1^t$, $s=1,\dots, S$. This approach may require a long burn-in for $\theta_1$ and necessitates many chains for sampling $\theta_2$, which may make evaluating convergence difficult. Generalizing to $J>2$ parameters requires exponentially more chains -- for example, sampling $\theta_3\mid \theta_1^t,\theta_2^s$ for all pairs previous of samples $\theta_1^t$ and $\theta_2^2$.}

\textcolor{black}{\cite{jacob2020unbiased} develop a general framework for unbiased integral estimation using coupled Markov chains, which can be used to speed up the naive approach just discussed. The key idea is to factor expectations under the sequential posterior using the law of iterated expectations and to estimate each of the iterated expectations with a telescoping sum. The telescoping sum is defined in terms of coupled Markov chains and depends almost surely on only a finite number of samples, drastically shortening the number of iterations required for each chain. Unfortunately, this approach also scales poorly to $J>2$ losses. For example, estimating expectations for $\theta_3\mid \theta_1,\theta_2$ requires Markov chains for $\theta_3\mid \theta_1^t,\theta_2^s$ with $s$, $t$ running over all possible pairs.}

\textcolor{black}{\cite{plummer2015cuts} develop an alternative solution based on the limiting distribution of an intentionally misspecified Metropolis-Hastings algorithm in which $\theta_1'\sim q_1(\cdot\mid\theta_1^t)$ is accepted based on the posterior ratio
\begin{align*}
     L_1(\theta_1'\mid \theta_1^t)&= \frac{\pi_1(\theta_1' \mid x)}{\pi_1(\theta_1^t\mid x)} 
\end{align*}
and proposal densities, and $\theta_2'\mid \theta_1^t\sim q_2(\cdot\mid \theta_1^t, \theta_2^t)$ is accepted based on the posterior ratio
\begin{align*}
     L_2(\theta_1'\mid \theta_1^t)&= \frac{\pi(\theta_2'\mid \theta_1^t, x)}{\pi(\theta_2^{t}\mid \theta_1^{t}, x)} 
\end{align*}
and proposal densities. The limiting distribution of samples from this algorithm is shown to be the weighted density $w(\theta_1, \theta_2)\pi(\theta_1,\theta_2\mid x)$, with
\begin{align*}
    w(\theta_1,\theta_2) &= \int \frac{\pi(\theta_2'\mid x, \theta_1')}{\pi(\theta_2' \mid x, \theta_1)}K_1(\theta_1\to\theta_1')K_2(\theta_2\to\theta_2'\mid\theta_1)d\theta_1'd\theta_2'.
\end{align*}
where $K_1$ and $K_2$ are induced transition kernels. \cite{plummer2015cuts} demonstrate that the weight function converges to $1$ under the double limit where (i) the transitions $\theta_1^t\to\theta_1^{t+1}$ only allow vanishingly small steps and (ii) the proposed value $\theta_2'$ does not depend on $\theta_2^t$. This leads to a path sampling algorithm which introduces intermediate values of $\theta_1$ to approximate a small step-size and reduce dependence between draws of $\theta_2$.}

\textcolor{black}{Iteration $t$ of the path sampling algorithm begins with samples $\theta_1^{t-1}$ and $\theta_1^{t}$ from the misspecified Metropolis-Hastings algorithm. These are used to define a discrete path with $m$ steps according to the equation
\begin{align*}
    \theta_1^{t-1,j} = (j/m)\theta_1^{t-1} + (1-j/m)\theta_1^{t},\quad j=1,...,m.
\end{align*}
A new value $\theta_2^t$ is generated by sampling along this path. Explicitly, one proposes/accepts $\theta_2^{t,j}\sim q_2(\cdot\mid\theta_1^{t,j},\theta_2^{t,j-1})$ for $j=1,...,m$ using the misspecified Metropolis-Hastings algorithm, and retains the final value $\theta_2^t = \theta_2^{t-1,m}$. \cite{plummer2015cuts} argue the weight function converges to $1$ in the limit $m\to\infty$, in which the steps between $\theta_1^{t-1,j}$ and $\theta_1^{t-1,j+1}$ become arbitrarily small and $\theta_1^t$ loses dependence on $\theta_1^{t-1}$.}

\textcolor{black}{A major advantage of the path sampling approach is that it can be generalized to $J>2$ losses with minimal additional computation by concatenating all previous parameters into a single vector -- for example, sampling $\theta_3^{t,j}\sim q_3(\cdot\mid\theta_1^t, \theta_2^t, \theta_3^{t,j-1})$ along the path $(\theta_1^{t,j}, \theta_2^{t,j}) = (j/m)(\theta_1^{t}, \theta_2^{t}) + (1-j/m)(\theta_1^{t+1}, \theta_2^{t+1})$. To our knowledge, this is the only solution which scales acceptably with the number of losses, making it the most promising algorithm when exact sampling is not feasible. For non-Euclidean parameters, paths must be induced on the underlying manifolds. This can be achieved by mapping to Euclidean space with a chart centered on the current state. }

\subsection{Automatic hyperparameter selection}
We provide additional details for the hyperparameter tuning algorithm proposed in Section $3$. Let $\hat r_b$ be the bootstrap estimate of the radius of a 95\% confidence ball and $r_g(\eta)$ the radius of a 95\% credible ball under a Gibbs posterior with parameter $\eta$. The goal is to find $\eta^\star$ with $r_g(\eta^\star)=\hat r_b$. In most cases $r_g(\eta)$ is not available analytically and we approximate it pointwise via Monte Carlo. Let $\hat r_g(\eta)$ denote any unbiased approximation of $r_g(\eta)$.

Stochastic approximation \citep{robbins1951stochastic} is popular algorithm for solving equations of the form $E[f(\theta)]=c$ where $\theta$ is a random variable, $f$ is a function, and $c$ is a constant. Inspired by \cite{martin2022direct}, we use this algorithm to approximately solve $\hat r_g(\eta) = \hat r_b$. This is formalized in Algorithm 2, which was used sequentially to calibrate the principal component analysis posterior in simulations and on the application to crime data. A summable step size, such as $\varepsilon_t=1/t^2$, is theoretically required to guarantee convergence, although in practice we find the algorithm converges faster with a non-summable step size such as $\varepsilon_t=1$ or $\varepsilon_t=1/t$, producing $\eta$ with $|\{\hat r_g(\eta_t)-\hat r_b\}/\hat r_b| < 0.01$ in 20-50 iterations.

\subsection{Crime Analysis}\label{sec:crime_sup}

\begin{figure}[t]
    \centering
    {\includegraphics[width=16cm]{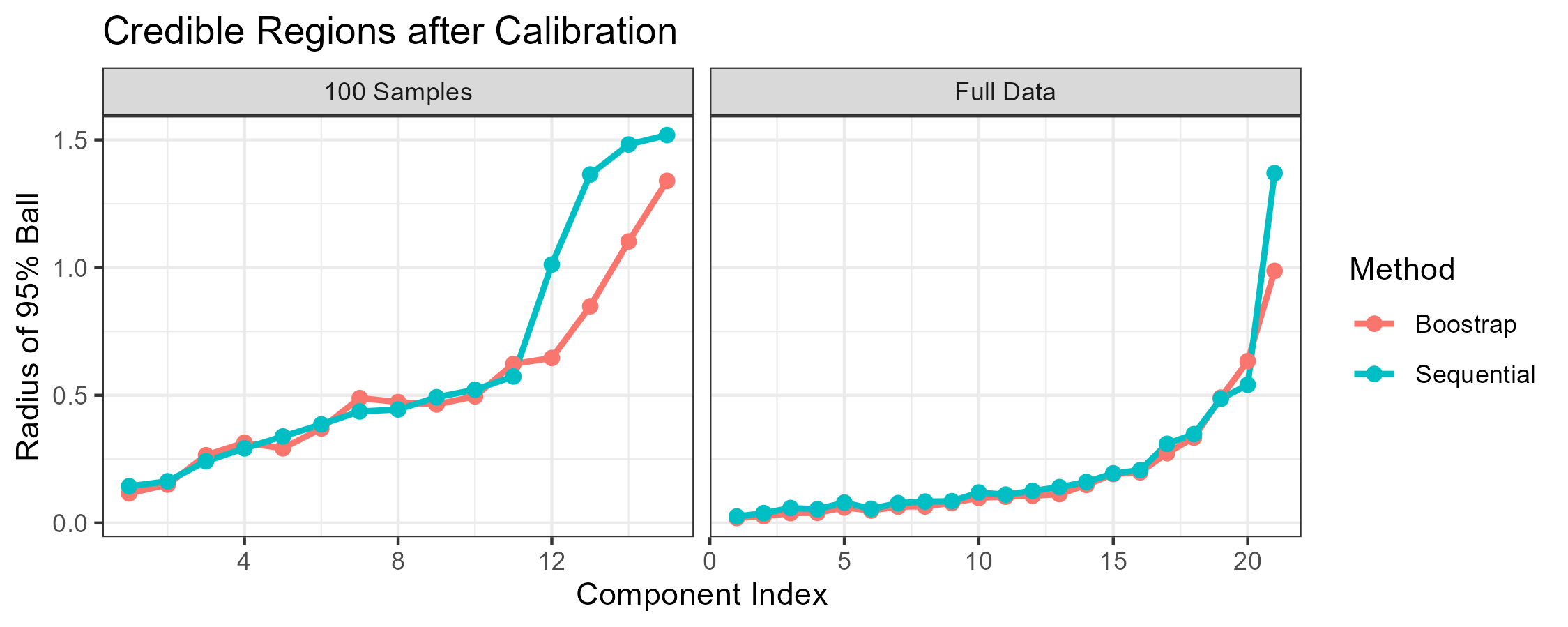}}
    \caption{Radii of credible/confidence balls around the maximum likelihood estimate for the sequential posterior and the bootstrap, after rotating all samples towards the maximum likelihood estimate.}%
    \label{fig:calibration}%
\end{figure}

The communities and crime dataset \citep{misc_communities_and_crime_183} contains $128$ features. For simplicity, we restrict analysis to the $p=99$ features which are available for all communities. Variable descriptions are available from the University of California, Irvine Machine Learning Repository at \url{https://archive.ics.uci.edu/dataset/183/communities+and+crime}. We begin by running loss-based principal component analysis on the full dataset of $n=1994$ communities after centering/scaling the features. The number of components required to explain $90\%$ of the variance is $k=21$. The $10$ variables with the largest absolute loadings on each component are presented in Tables \ref{table:vars1} and \ref{table:vars2}. We interpret the first five components as follows:
\begin{enumerate}
    \item \textbf{Income and Family Stability:} Variables related to median family income, median individual income, percentage of children in two-parent households, and percentage of households with investment/rent income.
    \item \textbf{Immigration and Language:} Variables related to recent immigration patterns, the percentage of the population that is foreign-born, and language proficiency.
    \item \textbf{Household Size and Urbanization:} Variables related to the number of persons per household, household size, and urban population.
    \item \textbf{Age and Stability:} Variables indicating the percentage of the population in older age groups and the percentage of people living in the same house or city as they did in previous years.
    \item \textbf{Population and Urban Density:} Variables related to population, land area, and overall community size.
    \item \textbf{Employment and Marital Status:} Variables related to relationship status, divorce rates, and advanced career progression.
\end{enumerate}

The remaining latent factors do not have clear interpretations, often simultaneously including variables related to family circumstance, age, race, and housing metrics. 

This dataset contains twenty times as many observations as features, hence credible regions for scores using the full data will be quite narrow. We subsample the data in order to better illustrate the nature of the uncertainty quantification provided our method. This is achieved by fitting $k$-means with $m=100$ clusters to the violent crime responses and then choosing a representative community within each cluster.

We use the calibration algorithm proposed in Section 3 to calibrate the Gibbs posterior on both the subsampled and full datasets. Specifically, the bootstrap was used to estimate the radius $\hat r_j$ of a $95\%$ confidence ball for $v_j$ centered at $\hat v_j$, and then Algorithm 2 was combined with Algorithm 1 to adjust $\eta_j$ so that a $95\%$ credible ball for $v_j$ had radius approximately $\hat r_j$. This was done sequentially: first for $\eta_1$, then $\eta_2$, and so on.

All samples were Procrustes aligned to the empirical loss minimizer prior to calculating intervals. We used a constant step size of $1$ and terminated calibration of the $j$th component whenever (1) the radius of the confidence/ball agreed to within $1\%$, (2) $\eta_j$ changed by less than $1\%$, or (3) the algorithm had run for 20 iterations. Figure \ref{fig:calibration} compares the Gibbs posterior balls to the bootstrap balls. In both cases, the first $10$ components are calibrated nearly exactly; and components thereafter have moderate errors. This is to be expected as the variance accumulates quickly with component index. More accurate calibration could be achieved by drawing more samples from the Gibbs posterior to reduce Monte Carlo error and running the stochastic approximation algorithm for more steps.

\subsection{Blueprint for Hierarchical Models}\label{sec:multilevel}

\textcolor{black}{The sequential Gibbs posterior can be used for a wide array of applications, including hierarchical models. We outline this construction for generalized linear regression with random effects. Let $x_{ik}$ and $y_{ik}$ be a covariate vector and response for sample $i=1,...,n_k$ in group $k=1,...,K$. It is common to infer group-specific coefficients $\theta_k$ and to share information between groups by shrinking to common coefficients $\theta_0$. Generalized linear models only require specification of the conditional mean and variance,
\begin{align*}
    E(Y_{ik}\mid x_{ik}) &= g^{-1}(x_{ik}^T\theta_k)\\ \text{var}(Y_{ik}\mid x_{ik})&=\gamma_k V\{g^{-1}(x_{ik}^T\theta_k)\}
\end{align*}
with $g$ a strictly monotone and differentiable link function, $V$ a positive and continuous function, and $\gamma_k$ a dispersion parameter. It is common to base inference for coefficients on the log-quasi-likelihood \citep{wedderburn1974quasi}:
\begin{align*}
    \ell_k^{(n_k)}(\theta_k\mid x) &= \frac{1}{\gamma_k}\sum_{i=1}^{n_k}\int_{0}^{g^{-1}(x_{ik}^T\theta_k)}\frac{y_{ik}-t}{V(t)}dt.
\end{align*}
The above loss allows for robust, semi-parametric inference for coefficients and often has appealing theoretical properties such as consistency \citep{chen1999strong}. The integrals are available in closed form for Gaussian, Poisson, and binomial distributions. \cite{agnoletto2025bayesian} define a Gibbs posterior using the log-quasi-likelihood as a loss, allowing for Bayesian uncertainty quantification in generalized linear models without group structure, subject to two mild moment assumptions. They derive a plug-in estimate for the dispersion parameter, avoiding the need to calibrate the Gibbs posterior, and demonstrate excellent coverage in realistic simulations.}

\textcolor{black}{It is straightforward to extend this construction to random effect models, for example by first estimating global coefficients $\theta_0$ using a loss $\ell_0^{(n)}(\theta_0\mid x)$ depending on all data (for example, a log-quasi-likelihood), and then estimating the group-specific coefficients using the log-quasi-likelihoods above, with an additional term that shrinks towards the global parameters:
\begin{align*}
    \ell_k^{(n_k)}(\theta_k\mid x, \theta_0) &= -\lambda||\theta_k - \theta_0||^2 + \frac{1}{\gamma_k}\sum_{i=1}^{n_k}\int_{0}^{g^{-1}(x_{ik}^T\theta_k)}\frac{y_{ik}-t}{V(t)}dt.
\end{align*}
The corresponding sequential posterior can be calibrated using existing plug-in estimates or our general calibration algorithm, and the shrinkage parameter $\lambda$ can be tuned by cross-validation. Sampling for general link functions is possible with the path algorithm detailed in Section S3.2.}

\newpage

\begin{table}
    \hspace*{-1.9cm}
    \centering
    \begin{tabular}{cl}
        \toprule
        \textbf{Latent Factor Index} & \textbf{Variables} \\
        \midrule
        1 & medFamInc, medIncome, PctKids2Par, pctWInvInc, PctPopUnderPov, \\
          & PctFam2Par, PctYoungKids2Par, perCapInc, pctWPubAsst, PctHousNoPhone \\
        \addlinespace
        2 & PctRecImmig10, PctRecImmig8, PctRecImmig5, PctRecentImmig, \\
          & PctForeignBorn, PctSpeakEnglOnly, PctNotSpeakEnglWell, \\
          & PctPersDenseHous, racePctAsian, racePctHisp \\
        \addlinespace
        3 & PersPerOccupHous, PersPerFam, PersPerOwnOccHous, householdsize, \\
          & PersPerRentOccHous, PctLargHouseOccup, HousVacant, PctLargHouseFam, \\
          & numbUrban, population \\
        \addlinespace
        4 & PctSameCity85, agePct12t29, PctSameHouse85, agePct16t24, \\
          & agePct12t21, agePct65up, pctWSocSec, PctImmigRec5, PctImmigRecent, \\
          & PctImmigRec8 \\
        \addlinespace
        5 & population, LandArea, numbUrban, NumUnderPov, NumIlleg, \\
          & HousVacant, NumInShelters, agePct65up, NumStreet, PctHousLess3BR \\
        \addlinespace
        6 & PctEmplProfServ, MalePctDivorce, TotalPctDiv, FemalePctDiv, \\
          & agePct12t21, MalePctNevMarr, MedYrHousBuilt, agePct16t24, \\
          & PctVacMore6Mos, PctEmploy \\
        \addlinespace
        7 & racepctblack, PctIlleg, PctEmplManu, PctEmploy, PctHousOccup, \\
          & PctWorkMomYoungKids, PctWorkMom, pctWWage, racePctWhite, \\
          & PctBornSameState \\
        \addlinespace
        8 & PctHousOccup, racepctblack, PctEmplManu, MedYrHousBuilt, PopDens, \\
          & racePctWhite, PctOccupManu, PctBornSameState, PersPerRentOccHous, \\
          & MedRentPctHousInc \\
        \addlinespace
        9 & PctImmigRec8, PctImmigRec5, PctImmigRec10, PctImmigRecent, \\
          & MalePctNevMarr, pctWFarmSelf, MedOwnCostPctInc, MedRentPctHousInc, \\
          & agePct12t29, agePct16t24 \\
        \addlinespace
        10 & PctWorkMomYoungKids, PctWorkMom, MedOwnCostPctIncNoMtg, pctWFarmSelf, \\
           & pctUrban, pctWRetire, PctHousOccup, PctVacMore6Mos, PctWOFullPlumb, \\
           & MedOwnCostPctInc \\
        \bottomrule
    \end{tabular}
    \caption{Latent factors 1-10 and their corresponding variables. Continued in Table \ref{table:vars2}}
    \label{table:vars1}
\end{table}

\begin{table}
    \hspace*{-1.9cm}
    \centering
    \begin{tabular}{cl}
        \toprule
        \textbf{Latent Factor Index} & \textbf{Variables} \\
        \midrule
        11 & PctWorkMom, pctWRetire, PctWorkMomYoungKids, PctHousOwnOcc, \\
           & MedRentPctHousInc, pctWSocSec, PctImmigRec5, agePct65up, \\
           & PctImmigRec8, racepctblack \\
        \addlinespace
        12 & MedOwnCostPctIncNoMtg, PctWorkMomYoungKids, PctWorkMom, \\
           & PctEmplProfServ, PctEmplManu, PctImmigRec5, pctWFarmSelf, \\
           & PctImmigRecent, racePctWhite, MedOwnCostPctInc \\
        \addlinespace
        13 & PctUsePubTrans, PctEmplManu, pctUrban, MedRentPctHousInc, \\
           & MedOwnCostPctIncNoMtg, OwnOccMedVal, OwnOccLowQuart, OwnOccHiQuart, \\
           & PopDens, PctSameState85 \\
        \addlinespace
        14 & PctSameState85, PctEmplManu, PctBornSameState, MedRentPctHousInc, \\
           & AsianPerCap, racepctblack, pctWFarmSelf, FemalePctDiv, TotalPctDiv, \\
           & MalePctDivorce \\
        \addlinespace
        15 & indianPerCap, LemasPctOfficDrugUn, pctUrban, PctUsePubTrans, \\
           & NumStreet, MedNumBR, racePctAsian, NumInShelters, LandArea, \\
           & PopDens \\
        \addlinespace
        16 & indianPerCap, LemasPctOfficDrugUn, pctUrban, NumStreet, \\
           & PctUsePubTrans, PopDens, LandArea, AsianPerCap, MedNumBR, \\
           & NumInShelters \\
        \addlinespace
        17 & PctVacMore6Mos, PctVacantBoarded, MedYrHousBuilt, racePctAsian, \\
           & indianPerCap, LemasPctOfficDrugUn, MedOwnCostPctIncNoMtg, pctWRetire, \\
           & MedOwnCostPctInc, AsianPerCap \\
        \addlinespace
        18 & LemasPctOfficDrugUn, pctWFarmSelf, AsianPerCap, HispPerCap, \\
           & MedNumBR, blackPerCap, PctVacantBoarded, PctHousOccup, \\
           & MedOwnCostPctIncNoMtg, PctWOFullPlumb \\
        \addlinespace
        19 & AsianPerCap, blackPerCap, PctWOFullPlumb, PctLess9thGrade, \\
           & HispPerCap, LemasPctOfficDrugUn, MedNumBR, pctUrban, \\
           & MedRentPctHousInc, NumStreet \\
        \addlinespace
        20 & racePctAsian, pctWFarmSelf, blackPerCap, LandArea, PctWOFullPlumb, \\
           & MedRentPctHousInc, PctLargHouseFam, MedNumBR, PctLargHouseOccup, \\
           & MedOwnCostPctIncNoMtg \\
        \addlinespace
        21 & PctWOFullPlumb, LemasPctOfficDrugUn, NumStreet, pctWFarmSelf, \\
           & racePctAsian, racePctWhite, racepctblack, MedYrHousBuilt, \\
           & PctSameState85, PctVacMore6Mos \\
        \bottomrule
    \end{tabular}
    \caption{Latent factors 11-21 and their corresponding variables.}
    \label{table:vars2}
\end{table}

\end{document}